\documentclass{article}

\usepackage{arxiv}

\usepackage[utf8]{inputenc} 
\usepackage[T1]{fontenc}    
\usepackage{hyperref}       
\usepackage{url}            
\usepackage{booktabs}       
\usepackage{amsfonts}       
\usepackage{nicefrac}       
\usepackage{microtype}      
\usepackage{lipsum}
\usepackage{amsmath}

\usepackage{amsfonts,amssymb,stmaryrd}
\usepackage{graphicx}
\usepackage{amsthm}
\usepackage{xcolor}

\newcommand{\R}{\mathbb{R}}

\newtheorem{theorem}{Theorem}[section]

\newtheorem{corollary}{Corollary}[section]

\newtheorem{remark}{Remark}[section]

\usepackage{subcaption}
\usepackage{cite}
\usepackage{authblk}
\newcommand{\p}{\partial}
\newcommand{\bb}{\begin{equation}}
\newcommand{\ee}{\end{equation}}
\newcommand{\ba}{\begin{array}}
\newcommand{\ea}{\end{array}}
\newcommand{\f}{\frac}
\usepackage[all]{xy}
\newcommand{\ds}{\displaystyle}

\newcommand{\al}{\alpha}
\newcommand{\be}{\beta}

\newcommand{\sign}{\text{sgn}\,}

\usepackage{extarrows}

\numberwithin{equation}{section}

\usepackage{ulem}

\usepackage[most]{tcolorbox}

\definecolor{TextColor}{rgb}{0.75, 0.75, 0.75}

\title{An integrable pseudospherical equation with pseudo-peakon solutions}

\author[1]{Priscila Leal da Silva}\author[2] {Igor Leite Freire}\author[3]{Nazime Sales Filho}

\affil[1,2]{Departamento de Matemática, Universidade Federal de São Carlos\\
Rodovia Washington Luís, Km 235, 13565-905\\
São Carlos-SP, Brasil\\
$^1$\texttt{priscilals@ufscar.br} 
  $^2$\texttt{igor.freire@ufscar.br} \\
  $^2$\texttt{igor.leite.freire@gmail.com}}
\affil[3]{Faculdade de Ciência e Tecnologia\\Universidade Federal de Mato Grosso\\Av. Fernando Corrêa da Costa, nº 2367, Bairro Boa Esperança, 78060-900\\ Cuiabá, MT - Brasil\\
\texttt{sales.nazime@gmail.com}}
 
\begin{document}
\maketitle
\begin{abstract}
We study an integrable equation whose solutions define a triad of one-forms describing a surface with Gaussian curvature -1. We identify a local group of diffeomorphisms that preserve these solutions and establish conserved quantities. From the symmetries, we obtain invariant solutions that provide explicit metrics for the surfaces. These solutions are unbounded and often appear in mirrored pairs. We introduce the ``collage'' method, which uses conserved quantities to remove unbounded parts and smoothly join the solutions, leading to weak solutions consistent with the conserved quantities. As a result we get pseudo-peakons, which are smoother than Camassa-Holm peakons. Additionally, we apply a Miura-type transformation to relate our equation to the Degasperis-Procesi equation, allowing us to recover peakon and shock-peakon solutions for it from the solutions of the other equation.

\end{abstract}


\keywords{Integrable equations \and Equations describing pseudospherical surfaces \and Symmetries \and Conserved quantities \and Blow up of solutions \and Shock-peakons}
\newpage

\section{Introduction}\label{sec1}

In this paper we consider the equation 
\bb\label{1.0.1}
u_t-u_{txx}=16uu_x-8u_xu_{xx}+2u_{xx}^2-4uu_{xxx}+2u_xu_{xxx},
\ee
where $u=u(x,t)$ is a field variable, while $t$ and $x$ can be regarded as temporal and spatial variables, respectively.

This equation was studied a couple of years ago in \cite{tu-dcds,tu-na,tu-mana} from a qualitative perspective. More precisely, these references showed that Cauchy problems involving \eqref{1.0.1} are well-posed in certain Besov spaces \cite[Theorem 3.1]{tu-na} whereas global weak solutions were considered in \cite{tu-mana}. Moreover, it also has solutions blowing up in finite time in the form of wave breaking \cite[Theorem 4.1]{tu-na}, see also \cite[Section 3]{tu-mana}.

Despite these results of more qualitative nature, the equation first appeared in the context of integrable systems, having been discovered by Novikov, see \cite[Equation 16]{nov}. It has an infinite hierarchy of symmetries \cite[Theorem 3]{nov}, and is Lax integrable in the sense that it is the compatibility condition for the system
\bb\label{1.0.2}
\ba{l}
\psi_x-\psi_{xxx}-\lambda(2m-m_x)\psi=0,\\
\\
\ds{\psi_t=\f{2}{\lambda}\psi_{xx}+2(2u-u_x)\psi_x-2\Big(2u_x-u_{xx}+\f{2}{3\lambda}\Big)\psi},
\ea
\ee
where $m=u-u_{xx}$.

System \eqref{1.0.2} implies \eqref{1.0.1} as the resulting compatibility condition of a $\frak{sl}(3,\R)$ valued zero curvature representation (ZCR).

While \eqref{1.0.2} is enough to prove the existence of a Lax pair for the equation, it is also a strong indication of that \eqref{1.0.1} could not describe two-dimensional objects emerging from the compatibility of $\frak{sl}(2,\R)-$valued ZCR, as observed in other integrable models like the famous Korteweg-de Vries or Camassa-Holm equations \cite{sasaki,chern,reyes2002,reyes2011}. Surprisingly, we have discovered that \eqref{1.0.1} can also be obtained from a $2\times 2$ ZCR, which in some sense violates what would be expected from an equation having a third order Lax pair like \eqref{1.0.2}.

In the next section we explore this unexpected geometric nature of the equation, exhibiting a triad of one-forms satisfying the structure equations for a pseudospherical surface, in the sense of the works by Chern and Tenenblat \cite{chern} and Reyes \cite{reyes2002,reyes2011}.

In Section \ref{sec3} we classify the characteristics for conservation laws of the equation up to second order. We were able to find five characteristics for conservation laws, three of them having second order terms. Next, we use them to establish conservation laws and then, conserved quantities.

Lie symmetries are obtained in Section \ref{sec4}. From them we obtain some solutions, that we explore in connection with the pseudospherical structure of the equation to provide explicit metrics for the corresponding surfaces.

We have noticed that some of the invariant solutions could be combined in order to give rise to new solutions compatible with some of the conservation laws. We have named this new procedure as {\it collage process}. In particular, such a process can also derive the famous peakon solutions for the Camassa-Hom and Degasperis-Procesi equation (this is explained in the Discussion). In our particular case, it has as an outcome the emergence of continuously differentiable peakon-shaping  solutions for the equation. 

Although these ``peakons'' are differentiable ($C^1$ regularity), they are not strong solutions for \eqref{1.0.1}, since their second derivatives have jumps. Remarkably, we also obtained a system of ODEs to what we might call ``2 continuously differentiable peakon-like'' solutions. These additional solutions leave invariant some of the conserved quantities we report in section 3. However, we also succeeded in finding non-conservative peakon like solutions, that unlike their conservative siblings, may blow up in finite time.  These results are shown in Section \ref{sec5}.

The solutions obtained in Section \ref{sec5} belong to the Sobolev space $H^{3/2+\epsilon}(\R)$, $\epsilon\in[0,1)$, as long as they exist. This motivates us to look for generalisations of them. In Section \ref{sec6} we report a class of conservative $C^1$ solutions and, more interestingly, we also obtained $C^1$ solutions blowing up in finite time.

In Section \ref{sec7} we look for a sort of multi-peakon analogous for the solutions we found in sections \ref{sec5} and \ref{sec6}. We obtained two types of solutions: those leaving invariant some functional and globally well-defined, and others developing singularities in finite time.

We use a linear operator between Sobolev spaces applying $H^s(\R)$ into $H^{s-1}(\R)$, and show relations between solutions of \eqref{1.0.1} and the famous Degasperis-Procesi (DP) equation
\bb\label{1.0.3}
u_t-u_{txx}+4uu_x=3u_xu_{xx}+uu_{xxx}.
\ee

As a result, from solutions of \eqref{1.0.1} we obtain solutions for the DP equation \eqref{1.0.3} in Section \ref{sec8}.

Our results are discussed in Section \ref{sec9}, while our conclusions are given in Section \ref{sec10}.

\section{Relations with pseudospherical surfaces}\label{sec2}

By a differential function we mean analytic functions depending on $(x,t,u)$ and derivatives of $u$ of finite order, but arbitrary. Here $u$ is assumed to be a function of $(x,t)$. Their collection is denoted by ${\cal A}$. We say that a differential function $f\in{\cal A}$ has order $k$, where $k$ is a non-negative integer, if the highest derivative appearing in $f$ is of order $k$. For the particular case whenever $k=0$, that is, $f=f(x,t,u)$, then $f$ is said to be a $0-$th order differential function. For further details, see \cite[Chapter 5]{olverbook} and \cite{vitolo}.

A differential equation for a function $u=u(x,t)$, ${\cal E}=0$, 
is said to describe a pseudospherical surface (PSS equation for short) if there are one-forms 
\bb\label{2.0.1}
\omega_i=f_{i1}dx+f_{i2}dt,\quad 1\leq i\leq 3,
\ee
where $f_{ij}\in{\cal A}$ and the triple $\{\omega_1,\omega_2,\omega_3\}$ satisfies 
\bb\label{2.0.2}
d\omega_1=\omega_3\wedge\omega_2,\quad d\omega_2=\omega_1\wedge\omega_3,\quad d\omega_3=\omega_1\wedge\omega_2
\ee
modulo ${\cal E}=0$. In consequence, regions (simply connected open sets) of the domain of the solution $u$ for which $(\omega_1\wedge\omega_2)(u(x,t))\neq0$ everywhere enable us to define a Riemannian metric
$$
g=\omega_1^2+\omega_2^2
$$
determining a surface of constant Gaussian curvature ${\cal K}=-1$. Surfaces of constant and negative Gaussian curvature are called {\it pseudospherical surfaces} (PSS). For further details on PSS equations, see \cite{chern,reyes2002,reyes2011,sasaki}.

We want to avoid problems regarding regularity of solutions when treated from a geometric viewpoint (for a better discussion, see \cite{freire-ch}). Therefore, throughout this paper, whenever we consider \eqref{1.0.1} in connection with PSS, we will always assume $C^\infty$ solutions.

Consider the one-forms
\bb\label{2.0.3}
\ba{lcl}
\omega_1&=&-2dx,\\
\\
\omega_2=\omega_3&=&\ds{\Big(1-2u+u_x+2u_{xx}-u_{xxx}\Big)dx}\\
\\
&+&\ds{\Big(16u_x^2-16uu_x+16uu_{xx}-16u_xu_{xx}-4uu_{xxx}+2u_{xx}^2+2u_xu_{xxx}\Big)dt}.
\ea
\ee

A lengthy, but simple, calculation shows that
\begin{align}\label{2.0.4}
    \begin{aligned}
        d\omega_1&-\omega_3\wedge\omega_2\equiv0,\\
    d\omega_2&-\omega_1\wedge\omega_3=d\omega_3-\omega_1\wedge\omega_2\\
    &=(2-\partial_x)\Big(u_t-u_{txx}-\big(16uu_x-8u_xu_{xx}+2u_{xx}^2-4uu_{xxx}+2u_xu_{xxx}\big)\Big)dx\wedge dt.
    \end{aligned}
\end{align}

Therefore, on the solutions of \eqref{1.0.1} the set of equations \eqref{2.0.4} is equivalent to \eqref{2.0.2}. This proves the following result.

\begin{theorem}\label{teo2.1}
    Equation \eqref{1.0.1} is the compatibility condition for the one-forms \eqref{2.0.3} to describe a PSS. 
\end{theorem}

Not all solutions of an equation lead to a PSS structure. Actually, in order for a solution to define a PSS structure, it has to satisfy the condition $(\omega_1\wedge\omega_2)(u(x,t))\neq0$ on the region $u$ is evaluated. Such solutions are called ${\it generic}$, whereas those for which $(\omega_1\wedge\omega_2)(u(x,t))=0$ are called non-generic solutions.

Let us characterise the non-generic solutions of \eqref{1.0.1}.

\begin{theorem}\label{teo2.2}
    Let $\Omega\subseteq\R^2$ be an open and simply connected set, and $u:\Omega\rightarrow\R$ be a solution of \eqref{1.0.1}. Then $u$ is a non-generic solution of \eqref{1.0.1} if and only if $u$ satisfies \eqref{1.0.1} and the ODE $\p_x(2-\p_x)(u_x-2u)^2=0$ as well.
\end{theorem}

\begin{proof}
A solution $u$ of \eqref{1.0.1} is a non-generic solution for \eqref{1.0.1} if and only if $\omega_1\wedge\omega_2\big|_\Omega\equiv0$. From \eqref{2.0.3} we conclude that
$$
\omega_1\wedge\omega_2=8(u_x-2u)(u_{xx}-2u_x)-4(u_x-2u)(u_{xxx}-2u_{xx})-(2u_{xx}-4u_x)^2,
$$
that, after rearranging, is equivalent to
\bb\label{2.0.5}
\omega_1\wedge\omega_2=2\p_x(2-\p_x)(u_x-2u)^2.
\ee

As a result, from \eqref{2.0.5} and \eqref{1.0.1} we conclude that $u$ is non-generic if and only if $u$ solves \eqref{1.0.1} and also
$\p_x(2-\p_x)(u_x-2u)^2=0$.
\end{proof}

In consequence of the above results, if $u$ is not a solution as described in Theorem \ref{teo2.2}, then
\bb\label{2.0.7}
\ba{lcl}
g&=&\Big(4+\Big(1-2u+u_x+2u_{xx}-u_{xxx}\Big)^2\Big)dx^2+2\ds{\Big(1-2u+u_x+2u_{xx}-u_{xxx}\Big)\times}\\
\\
&&\ds{\times\Big(16u_x^2-16uu_x+16uu_{xx}-16u_xu_{xx}-4uu_{xxx}+2u_{xx}^2+2u_xu_{xxx}\Big)dxdt}\\
\\
&+&\ds{\Big(16u_x^2-16uu_x+16uu_{xx}-16u_xu_{xx}-4uu_{xxx}+2u_{xx}^2+2u_xu_{xxx}\Big)^2dt^2}
\ea
\ee
defines a metric on any open and simply connected subset of the domain of $u$ in which it does not agree with the solution characterised in Theorem \ref{teo2.2} at any point.

\section{Conserved quantities}\label{sec3}

Let ${\cal E}={\cal E}(t,x,u,\cdots)$, $C^0,C^1$ and $\phi$ be differential functions satisfying the relation
\bb\label{3.0.1}
D_tC^0+D_x C^1=\phi {\cal E}.
\ee

On the solutions of the differential equation ${\cal E}=0$ we have $D_t C^0+D_x C^1\equiv0$. We then say that the pair $(C^0,C^1)$ is a conserved current for the PDE ${\cal E}=0$, whereas the function $\phi$ in \eqref{3.0.1} is called {\it characteristic} of the conservation law (characteristic, for short). In particular, \eqref{3.0.1} is called {\it characteristic form of the conservation laws}. For further details, see \cite[page 266]{olverbook}.

Very often, we know the equation, but not necessarily its conserved currents nor the corresponding characteristics. Usually we can determine the characteristics of the conservation law using the fact that total derivatives belong to the kernel of the Euler-Lagrange operator \cite[Theorem 4.7, page 248]{olverbook}. Finding these differential functions is, definitively, not an easy task but, fortunately, there are some symbolic packages available at our disposal enabling us to obtain characteristics for an equation algorithmically once the order of the characteristic is fixed, see \cite{chev-2007,chev-2010-1,chev-2010-2, chev-2014, chev-2017}. With their help, we have the following result.

\begin{theorem}\label{teo3.1}
    Up to second order, the characteristics of any conservation law for \eqref{1.0.1} are
    \bb\label{3.0.2}
    \ba{lcl}
	\phi_1&=&1,\quad \phi_2=e^{2x},\quad\phi_3=u,\\
 \\
	\phi_4&=&-4u^{2}+u_{x}^{2}+2uu_{xx}-u_{x}u_{xx}-\dfrac{1}{2}u_{tx},\\
 \\
	\phi_5&=&e^{-2x}\left(3u^{2}-4uu_{x}-u_{x}^{2}+u_{t}-4uu_{xx}+2u_{x}u_{xx}+u_{tx} \right).
	\ea
	  \ee
\end{theorem}

In consequence, we have the following conservation laws for \eqref{1.0.1}:

\begin{theorem}\label{teo3.2}
    The corresponding conservation laws for \eqref{1.0.1}, determined by the characteristics given by Theorem \ref{teo3.1}, are $D_tC^0+D_x C^1=0$, where the components $C^0$ and $C^1$ of the conserved current $(C^0,C^1)$ are, given by

    \begin{enumerate}
		\item For $\phi_1=1$, we have
		\begin{equation*}
		\left.\begin{aligned}
        C_1^0&=u,\\
		\\
		C_1^1&=-8u^{2}+2u_{x}^{2}+4uu_{xx}-2u_{x}u_{xx}-u_{tx}.
		\end{aligned}\right.
		\end{equation*}

		\item For $\phi_2=e^{2x}$, we have 
		\begin{equation*}
		\left.\begin{aligned}
		C_2^0&=e^{2x}u,\\
		\\
		C_2^1&=\f{2}{3}e^{2x}\Big(4uu_{x}-2uu_{xx}-2u_{x}^2+u_{x}u_{xx}+\dfrac{1}{2}u_{tx}-u_{t}\Big).
		\end{aligned}\right.
		\end{equation*}

  \item For $\phi_3=u$, we have
		\begin{equation*}
		\left.\begin{aligned}
		C_3^0&=\dfrac{1}{2}u^2+\dfrac{1}{2}u_{x}^2,\\
		\\
		C_3^1&=-\dfrac{16}{3}u^{3}-2uu_{x}u_{xx}+\dfrac{2}{3}u_{x}^{3}+4u^2u_{xx}-uu_{tx}.
		\end{aligned}\right.
		\end{equation*}
		
		\item For $\phi_4=-4u^{2}+u_{x}^{2}+2uu_{xx}-u_{x}u_{xx}-\dfrac{1}{2}u_{tx}$, we have 
		\begin{equation*}
		\left.\begin{aligned}
		C_4^0&=-\dfrac{4}{3}u^3+\dfrac{1}{6}u_{x}^3-uu_{x}^2,\\
		\\
		C_4^1&=16u^4-8u^2u_{x}^2-2uu_{tx}u_{xx}-\dfrac{1}{4}u_{t}^2-16u^3u_{xx}-2u_{x}^3u_{xx}-\dfrac{1}{2}u_{x}^2u_{t}+4u^2u_{xx}^2\\
		&+u_{x}^2u_{xx}^2+\dfrac{1}{4}u_{tx}^2+8u^2u_{x}u_{xx}-4uu_{x}u_{xx}^2+4uu_{x}^2u_{xx}+2uu_{x}u_{t}+u_{x}u_{xx}u_{tx}+u_{x}^4\\
        &+4u^2u_{tx}-u^2_{x}u_{tx}.
		\end{aligned}\right.
		\end{equation*}
		
		\item For $\phi_5=e^{-2x}\left(3u^{2}-4uu_{x}-u_{x}^{2}+u_{t}-4uu_{xx}+2u_{x}u_{xx}+u_{tx} \right)$, we have 
		\begin{equation*}
		\left.\begin{aligned}
		C_5^0&=-3e^{-2x}\Big(u^3+uu_{x}^2+\dfrac{1}{3}u_{x}^3\Big),\\
		\\
		C_5^1&=\dfrac{1}{2}e^{-2x}(48u^3u_{x}+24u^3u_{xx}+8u^2u_{x}^2-44u^2u_{x}u_{xx}-16u^2u_{xx}^2-8uu_{x}^3+8uu_{x}^2u_{xx}\\
		&+16uu_{x}u_{xx}^2-4u_{x}^4+4u_{x}^3u_{xx}-4u_{x}^2u_{xx}^2+8uu_{x}u_{t}+8uu_{t}u_{xx}+8uu_{xx}u_{tx}+2u_{x}^2u_{t}\\
		&-4u_{x}u_{t}u_{xx}-4u_{x}u_{xx}u_{tx}-u_{t}^2-2u_{t}u_{tx}-u_{tx}^2+8uu_{x}u_{tx}-6u^2u_{tx}-12u^2u_{t}+2u^2_{x}u_{tx}).
		\end{aligned}\right.
		\end{equation*}
	\end{enumerate}
\end{theorem}

Equation \eqref{1.0.1} can be seen as a non-local evolution equation (when its solutions are restricted to certain function spaces) and, as we have already mentioned, its variable $t$ can be interpreted as time. The conservation laws are then particularly useful for the construction of quantities that are conserved on the solutions of the equation as long as they exist. Its worth pointing out that, differently from a conservation law, that is an intrinsic property of an equation, a conserved quantity is a property of the solution.

Integrating \eqref{3.0.1} with respect to $x$, and commuting time derivatives with spatial integration, we get
$$
\f{d}{dt}\int_\R C^0dx+C^1\big|_{-\infty}^{+\infty}=\int_\R\phi {\cal E}dx.
$$

If $u$ is a solution of ${\cal E}=0$ for which $C^{1}$ vanishes at infinity, we conclude that
$$
\f{d}{dt}\int_\R C^0dx=0,
$$
meaning that the quantity
$$t\mapsto\int_\R C^0dx$$
is time invariant, or in other words, it is a {\it conserved quantity}.

\begin{theorem}\label{teo3.3}
    Suppose that $u$ is a solution of \eqref{1.0.1} such that $u,u_x\rightarrow0$ as $|x|\rightarrow\infty$ and its second order derivatives are bounded. Then
    \bb\label{3.0.3}
    {\cal H}_0(t)=\int_{\mathbb{R}}u(x,t)\,dx,
    \ee
    \bb\label{3.0.4}
    \mathcal{H}_1(t)=\int_{\mathbb{R}}e^{2x}u(x,t)\,dx,
    \ee
    \bb\label{3.0.5}
 \mathcal{H}_2(t)=\dfrac{1}{2}\int_{\mathbb{R}}\Big(u^2+u_{x}^2\Big)(x,t)\,dx,
    \ee
    \bb\label{3.0.6}
    \mathcal{H}_3(t)=\int_{\mathbb{R}}\Big(-\dfrac{4}{3}u^3+\dfrac{1}{6}u_{x}^3-uu_{x}^2\Big)(x,t)\,dx
\ee
and
\bb\label{3.0.7}
    \mathcal{H}_4(t)=\int_{\mathbb{R}}e^{-2x}\Big(u^3+uu_{x}^2+\dfrac{1}{3}u_{x}^3\Big)(x,t)\,dx,
\ee 
are conserved quantities for the equation obtained from the conservation laws with characteristics $\phi_i$, $1\leq i\leq 5$, respectively. 

For the quantities \eqref{3.0.4} and \eqref{3.0.7} it is additionally assumed that $e^{|x|}u$, $e^{|x|}u_x$, $e^{2|x|}u_{t}$, $e^{2|x|}u_{tx}{\rightarrow 0}$ as $|x|\rightarrow\infty$.
\end{theorem}

The characteristics presented in Theorem \ref{teo3.1} were first obtained using the packages \cite{chev-2007,chev-2010-1,chev-2010-2, chev-2014, chev-2017}. They can be easily, but tediously, proved applying the Euler-Lagrange operator to \eqref{3.0.2} and taking \eqref{1.0.1} into account. It is equally straightforward to check that the conserved currents, their corresponding characteristics and \eqref{1.0.1} satisfy \eqref{3.0.1}, what concludes the proof of Theorem \ref{teo3.2}. Theorem \ref{teo3.3} is an immediate consequence of Theorem \ref{teo3.2}.

\begin{remark}\label{rem3.5}
    Let $-\infty\leq a<b\leq\infty$. The Sobolev space $H^1(a,b)$ is the set of distributions $f$ such that
    $$
    \int_a^b (f(x)^2+f'(x)^2)dx<\infty,
    $$
    the derivative above being considered in the distributional sense. Whenever $a=-\infty$ and $b=\infty$ we simply write $H^1(\R)$, that is endowed with the norm
    $$\|f\|_{H^1(a,b)}=\sqrt{\int_a^b (f(x)^2+f'(x)^2)dx}.$$
    From Theorem \ref{teo3.3}, if $u$ is a solution of \eqref{1.0.1} such that \eqref{3.0.5} is conserved, then $u(\cdot,t)\in H^1(\R)$.
\end{remark}

\section{Lie symmetries, reductions and explicit metrics}\label{sec4}

In this section we find Lie symmetries of \eqref{1.0.1}. Similarly to the characteristics, finding symmetries involves time consuming calculations (for further details, see \cite{olverbook}). Fortunately, likewise the characteristics, we have at our hands some packages \cite{stelios1,stelios2} that make easier the calculations for the generators. For this reason, below we present them without further discussion.

Using \cite{stelios1,stelios2} for finding the local group of diffeomorphism leaving invariant solutions of the equation, we can find the following generators
	\begin{equation}\label{4.0.1}
	X_1=\dfrac{\partial}{\partial x},\; X_2=\dfrac{\partial}{\partial t},\; X_3=t\dfrac{\partial}{\partial t}-u\dfrac{\partial}{\partial u},\; 
	X_4=e^{2x}\dfrac{\partial}{\partial u}.
	\end{equation}
 
The Lie algebra determined by \eqref{4.0.1} is summarised in Table \ref{tab1}.
 
 \begin{table}[h]
		\caption{Commutator table of the generators \eqref{4.0.1}}\label{tab1}
		\begin{center}
			\begin{tabular}{c|cccc}
				\toprule
				$[X_i,X_j]$	& $X_1$ & $X_2$ & $X_3$ & $X_4$ \\
				\midrule
				$X_1$ & $0$ &  $0$ &  $0$ &  $2X_4$ \\
				$X_2$ & $0$  & $0$  & $X_2$ &  $0$ \\
				$X_3$ & $0$ & -$X_2$ & $0$ &  $X_4$\\
				$X_4$ & $-2X_4$ & $0$ & $-X_4$ & $0$  \\
				\bottomrule
			\end{tabular}
		\end{center}
	\end{table}

Below we summarise the corresponding local group of diffeomorphisms determined by the generators \eqref{4.0.1} and their corresponding solutions obtained from a known solution $u(x,t)=f(x,t)$.

 \begin{table}[h]
	\caption{Summary of local transformations and corresponding new solutions. Below $\varepsilon\in\R$ is an arbitrary parameter and $u=f(x,t)$ is a known solution.}\label{tab2}
	\begin{center}
		\begin{tabular}{c|c|c}
			\toprule
			Generator	& Transformation & New solution \\
			\midrule
  $X_1$ & $(x+\varepsilon,t,u)$ &  ${\bar u}(x,t)=f(x-\varepsilon,t)$  \\
		$X_2$ & $(x,t+\varepsilon,u)$ &  ${\bar u}(x,t)=f(x,t-\varepsilon)$ \\
  $X_3$ & $(x,e^{\varepsilon} t,e^{-\varepsilon} u)$ &  ${\bar u}(x,t)=e^{\varepsilon} f(x,e^{\varepsilon} t)$\\
		$X_4$ & $(x,t,u+\varepsilon e^{2x})$ &  ${\bar u}(x,t)=f(x,t)-\varepsilon e^{2x}$   \\
			\bottomrule
		\end{tabular}
	\end{center}
\end{table}

\subsection{Reductions}

Table \ref{tab2} shows how we can obtain a new solution from a known one through the fluxes determined by the generators \eqref{4.0.1}. However, we can also find solutions from the generators using the following scheme: if $X$ denotes a linear combination of \eqref{4.0.1}, we can then obtain a function $u=\theta(x,t)$ invariant under $X$, in the sense that $X(u-\theta(x,t))=0$ whenever $u=\theta(x,t)$. We then require this function to be a solution for \eqref{1.0.1}. The main problem of such a process is the fact that we have an infinite number of possible linear combinations of the generators \eqref{4.0.1}.

Applying the procedure described in \cite[section 3.3]{olverbook}, among the arbitrary number of possibilities of linear combination of the generators \eqref{4.0.1} to seek invariant solutions, we can restrict ourselves to those given by 
\begin{equation}\label{4.1.1}
  \al X_1+\be X_3+X_4,\, \al X_1+X_3,\,\al X_1+X_2\,\, \mbox{e}\,\, X_1,
\end{equation}
where $\al$ and $\be$ are arbitrary parameters.

 \begin{table}[h]
	\caption{Adjoint representation. Below $\varepsilon$ is a parameter.}\label{tab3}
	\begin{center}
		\begin{tabular}{c|cccc}
			\toprule
			$Ad(e^{\varepsilon X_i})X_j$	& $X_1$ & $X_2$ & $X_3$ & $X_4$ \\
			\midrule
  $X_1$ & $X_1$ &  $X_2$ &  $X_3$ &  $2e^{-2\varepsilon}X_4$ \\
		$X_2$ & $X_1$  & $X_2$ & $X_3-\varepsilon X_2$ &  $X_4$ \\
  $X_3$ & $X_1$  & $e^{\varepsilon}X_2$ & $X_3$ & $e^{-\varepsilon}X_4$\\
		$X_4$ & $X_1+2\varepsilon X_4$ & $X_2$ & $X_3+\varepsilon X_4$ & $X_4$ \\
			\bottomrule
		\end{tabular}
	\end{center}
\end{table}

Using the invariance condition, we can transform the original PDE \eqref{1.0.1} into an ODE for a new variable $\theta$. Below we present schematically the ODE satisfied by the new dependent variable $\theta$ and its relation with the original variable $u$:

\begin{itemize}
    \item For the generator $\al X_1+\be X_3+X_4$, the corresponding ODE is
    \bb\label{4.1.2}
    \ba{l}
\al^2(\al^2-4\be^2)\theta'-\al^2\be^2z(5\theta''+z\theta''')+4\be(4\al^3-3\al\be^2 -\beta^3)\theta^2\\
\\
-8\be^3(3\al+4\be)z^2(\theta')^2-2\be^4z^4(\theta'')^2+4\be(4\al^3-15\al\be^2-7\be^3)z\theta\theta'\\
\\
-16\be^3(2\al+\be)z^2\theta\theta''-2\be^3(2\al+\be)z^3\theta\theta'''-8\be^3(\al+3\be)z^3\theta'\theta''\\
\\
-2\be^4z^4\theta'\theta'''=0,
\ea
\ee
where $\theta=\theta(z)$, $z=te^{-\be x/\al}$ and
\bb\label{4.1.3}
u(x,t)=e^{-\be x/\al}\theta(z)+\f{e^{2x}}{2\al+\be};
\ee

    \item For the generator $\al X_1+X_3$, the corresponding ODE is
    \bb\label{4.1.4}
\ba{l}
\al^2(\al^2-4)\theta'-\al^2z(5\theta''+z\theta''')+4(4\al^3-3\al-1)\theta^2-8({-3\al}+4)z^2(\theta')^2\\
\\
-2z^4(\theta'')^2+4(4\al^3-15\al-7)z\theta\theta'-16(2\al+1)z^2\theta\theta''-8(\al+3)z^3\theta'\theta''\\
\\
-2(2\al+1)z^3\theta\theta'''-2z^4\theta'\theta'''=0,
\ea
    \ee
where $\theta=\theta(z)$, $z=te^{-x/\al}$ and
\bb\label{4.1.5}
u(x,t)=e^{- x/\al}\theta(z);
\ee    

    \item For the generator $c X_1+X_2$ (note we have replaced $\al$ by $c$), the corresponding ODE is
    \bb\label{4.1.6}
-c(\theta'-\theta''')-16\theta\theta'+8\theta'\theta''-2(\theta'')^2+4\theta\theta'''-2\theta'\theta'''=0,
    \ee
where $\theta=\theta(z)$, $z=x-c t$ and
\bb\label{4.1.7}
u(x,t)=\theta(z);
\ee

\item For the generator $X_1$, the corresponding ODE the trivial solution $u(x,t)=c$, where $c$ is an arbitrary constant.
\end{itemize}

\subsection{Explicit solutions and their corresponding metrics}

Even though the solutions of the ODEs in the previous subsection led to solutions of \eqref{1.0.1}, they are still quite complicated to be solved in general. We could try to use symmetries again and obtain simpler ODEs to be solved, and then return to the original ODE to obtain its solution, leading to the corresponding solution for \eqref{1.0.1}. This would make this a longer paper focused on solutions obtained from symmetries, that although valid, is not our purpose. Our focus is not in proceeding with extensive classifications of solutions obtained from symmetries neither are the invariant solutions our main goal. 

The symmetries and their corresponding solutions have, however, a twofold importance for us: despite the problems mentioned above, we are still able to obtain some solutions in a fairly simple way from what we have obtained sor far. These solutions then enables us to find explicit metrics given by \eqref{2.0.7}. We are about to do this. The second relevant fact is that, combined with the conserved quantities we established in the precedent section, the invariant solutions give us clues to find unexpected but remarkable weak solutions for \eqref{1.0.1}. This is the subject of the coming section.

\subsubsection{Solutions from \eqref{4.1.2} and \eqref{4.1.4}}

If we substitute $\theta(z)=az$ into \eqref{4.1.2} we obtain a solution provided that $\al=2\be$. As a result, by \eqref{4.1.3} we have
\bb\label{4.2.1}
u(x,t)=ate^{-x}+\f{e^{2x}}{5\be},
\ee
where $a\in\R$ and $\be\neq0$ are arbitrary constants. Using the last transformation shown on Table \ref{tab2} we can remove the term $e^{2x}/(5\be)$ and get a simpler solution
\bb\label{4.2.2}
u(x,t)=ate^{-x}
\ee
to \eqref{1.0.1}.

According to Theorem \ref{teo2.2}, both \eqref{4.2.1} and \eqref{4.2.2} are generic solutions for \eqref{1.0.1}. Choosing $a=1/\sqrt{72}$, from solution \eqref{4.2.2} we obtain the following first fundamental form
$g=5dx^2+2t^2 e^{-2x}dxdt+t^4e^{-4x}dt^2$ for the region $t\neq0$.

\subsubsection{Solutions from \eqref{4.1.6}}

Equation \eqref{4.1.6} can alternatively be rewritten as
\bb\label{4.2.3}
\Big(-c(\theta-\theta'')-8\theta^2+2(\theta')^2-2\theta'\theta''+4\theta\theta''\Big)'=0,
\ee
that, once integrated, yields
\bb\label{4.2.4}
-c(\theta-\theta'')-8\theta^2+2(\theta')^2-2\theta'\theta''+4\theta\theta''=c_1,
\ee
where $c_1$ is an arbitrary constant. For $c_1=0$ we can obtain two solutions, namely,
\bb\label{4.2.5}
\theta_1(z)=\al_1 e^{-z}
\ee
and
\bb\label{4.2.6}
\theta_2(z)=-\al_2 e^{2z}+\f{\sqrt{-3c\al_2}}{2}e^z.
\ee

Although the $-$ sign in \eqref{4.2.6} is not mandatory, it is convenient due to coming calculations, so that we will work with it in that way.

Solution \eqref{4.2.6} is meaningless as long as $c\al_2>0$. For this reason, henceforth we only consider the constants satisfying the constraint $c\al_2\leq0$.

A straightforward calculation shows that
$\|\theta_1\|_{H^{1}(a,b)}=|\al_1|\sqrt{e^{-2a}-e^{-2b}},$
meaning that $\theta_1\in H^1(a,\infty)$, for any $a\in\R$, but $\theta_1\notin H^1(\R)$. Similarly, we have $\theta_2\in H^1(-\infty,b)$, for any $b\in\R$, but $\theta_2\notin H^1(\R)$. Observe, however, that both $\theta_1$ and $\theta_2$ are members of $H^1_{loc}(\R)$.

The corresponding solutions for \eqref{1.0.1} obtained from \eqref{4.2.5} and \eqref{4.2.6} are, respectively, given by
\bb\label{4.2.7}
u_1(x,t)=\al_1 e^{ct-x}
\ee
and
\bb\label{4.2.8}
u_2(x,t)=-\al_2 e^{2(x-ct)}+\f{\sqrt{-3c\al_2}}{2}e^{x-ct}.
\ee

While \eqref{4.2.7} is a generic solution for $\alpha_1\neq 0$, \eqref{4.2.8} is non-generic meaning that the $1$-forms $\omega_1$ and $\omega_2$ are such that $\omega_1\wedge \omega_2=0$. For \eqref{4.2.7}, let us exemplify the metric determined by it. By choosing $\al_1=1/\sqrt{72}$, we get
$
g=5dx^2+2 e^{2(ct-x)}dxdt+e^{4(ct-x)}dt^2.
$

\begin{remark}\label{rem4.1}
    Taking into account the comments above regarding the Sobolev norms of \eqref{4.2.5} and \eqref{4.2.6}, we conclude that $t\mapsto \|u_i(\cdot,t)\|_{H^1(a,b)}$ is not constant for any $a,b\in\R$, $i=1,2$.
\end{remark}

\section{The collage process and a new conservative solution}\label{sec5}

The solutions we found in the previous section are enough to provide us concrete examples for metrics \eqref{2.0.7} of PSS. Unfortunately, they are not compatible with the conserved quantities \eqref{3.0.3}--\eqref{3.0.7} since none of them belong to $L^\infty(\R)$ for fixed $t$.

\begin{remark}\label{rem5.1}
    From Remark \ref{rem4.1} we could have piecewise bounded functions locally behaving like solutions if we suitably cut off solutions \eqref{4.2.7} and \eqref{4.2.8}, namely, 
$$
\ba{lclccclcl}
\tilde{u}_1(x,t)&=&\left\{\ba{ll}u_1(x,t),&\text{ for }x>ct,\\
\\
0,&\text{ otherwise, }
\ea\right.
&&
\tilde{u}_2(x,t)&=&\left\{\ba{ll}u_2(x,t),&\text{ for }x<ct,\\
\\
0,&\text{ otherwise,}
\ea\right.
\ea
$$
but the paid price would be an eventual loss of continuity of these functions along the straight line $x=ct$. Therefore, although serving as solutions on some open sets of $\R^2$, they cannot be regarded as such on $\R^2$.
\end{remark}

Returning our eyes back to \eqref{4.2.3}--\eqref{4.2.6} and the reported sequel, our results were only obtained by setting the constant of integration $c_1$ as being $0$. We can, however, provide a better justification for that choice by imposing that $u$ and its derivatives vanish as $|x|\to \infty$. Although we have not succeeded in finding solutions vanishing for both $x\rightarrow\infty$ and $x\rightarrow-\infty$. On the contrary, we only obtained unbounded solutions, as previously mentioned.

Solution \eqref{4.2.5} is unbounded for $z\rightarrow-\infty$, but vanishes for $z\rightarrow+\infty$, whereas \eqref{4.2.5} has a complementary behaviour. Moreover, suitably calibrating the constants $\al_1$ and $\al_2$ we can make them agree at $z=0$, providing us a continuous function that solves the equation for both $z<0$ and $z>0$ regions, but potentially having some issues at $z=0$. Such a function cannot be a strong solution for (see equation \eqref{4.2.4})
\bb\label{5.0.1}
-c(\theta-\theta'')-8\theta^2-2(\theta')^2-(\theta'^2)'+2(\theta^2)''=0,
\ee
but we may have some hope that it could serve as a (conservative) weak one taking remarks \ref{rem4.1} and \ref{rem5.1} into account.

Multiplying \eqref{5.0.1} by $\varphi\in C^\infty_0(\R)$, where $C^\infty_0(\R)$ denotes the set of $C^\infty(\R)$ compactly supported functions and proceeding with integration by parts, we obtain the weak formulation for \eqref{5.0.1}:
\bb\label{5.0.2}
-\int_\R\Big(c\theta+8\theta^2+2 (\theta')^2\Big)\varphi dz+\int_\R(\theta')^2\varphi' dz+\int_\R\Big(c\theta+2\theta^2\Big)\varphi''dz=0.
\ee

A straightforward calculation shows that if $\theta\in C^3(\R)$ is a solution for \eqref{5.0.1}, then \eqref{5.0.2} is satisfied for any $\varphi\in C_0^\infty(\R)$, showing that any strong solution for \eqref{4.2.3} is a weak solution for \eqref{5.0.1}.

The way we cut the functions in Remark \ref{rem5.1} is not usually compatible with continuity (except for trivial solutions, that are, of course, out of our interest), but those new functions give us some clues about how to proceed to look for at least a continuous function from those two we already know.

Let $H(x)$ be the Heaviside step function, that is, 
\bb\label{5.0.3}
H(x)=\left\{
\ba{lcl}
1,&\text{if}& x>0,\\
\\
\ds{\f{1}{2}},&\text{if}& x=0,\\
\\
0,&\text{if}& x<0,
\ea
\right.
\ee
and consider the functions \eqref{4.2.5}--\eqref{4.2.6}. They have three constants involved, being $c$ the one corresponding to the wave speed, see \eqref{4.2.7} and \eqref{4.2.8}. Then we fix it, but we still have two degrees of freedom, that we reduce to one by choosing $\al_2$ compatible with the constraint $c\al_2\leq 0$.
\begin{figure}[h!]
	\centering
	\begin{subfigure}{0.45\linewidth}
		\includegraphics[width=\linewidth]{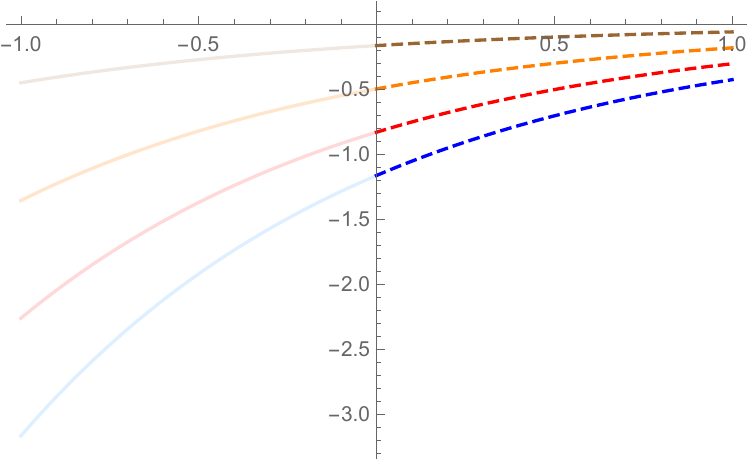}
		\caption{Function \eqref{4.2.5} with $\al_1=-c/6$ for different values of $c>0$. The light solid part of the function becomes unbounded as $z\rightarrow-\infty$, whereas the dashed part remains bounded.}
		\label{subfigA}
	\end{subfigure}\quad\quad\quad
	\begin{subfigure}{0.45\linewidth}
		\includegraphics[width=\linewidth]{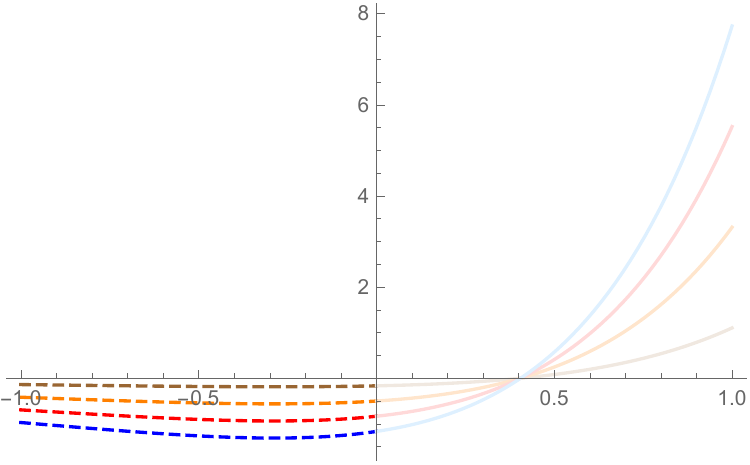}
		\caption{Function \eqref{4.2.6} with $\al_2=-c/3$, for different values of $c>0$. The light solid part of the function becomes unbounded as $z\rightarrow+\infty$, whereas the dashed part is bounded.}
		\label{subfigB}
	\end{subfigure}
	\begin{subfigure}{0.65\linewidth}
	        \includegraphics[width=\linewidth]{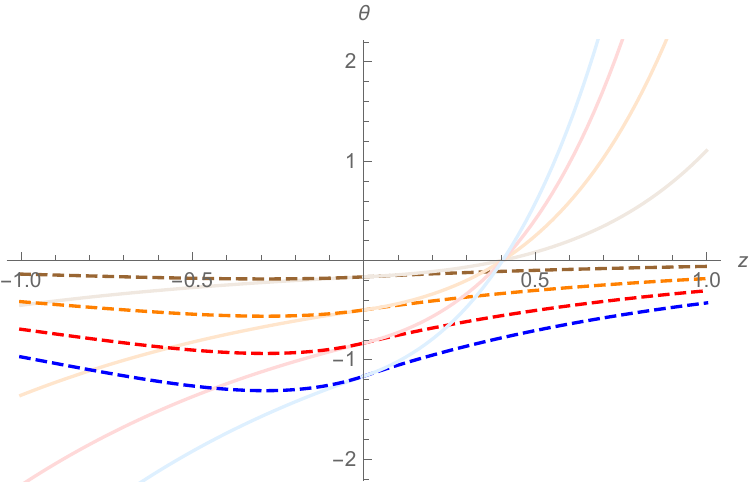}
	        \caption{A surgery fusing two solutions giving rise to a new function: we divide the solutions with respect to $z=0$ and discard their unbounded parts. The two remained bounded parts are continuously glued at $z=0$ using the weak formulation. The result is a bounded and continuous solution of the weak formulation \eqref{5.0.2}.}
	        \label{subfigC}
         \end{subfigure}
	\caption{Figure \ref{subfigC} shows how the collage process eliminates unbounded parts in a continuous process and produces a new, continuous and bounded, weak solution for \eqref{5.0.1} from the classical and unbounded functions \eqref{4.2.5} and \eqref{4.2.6}.}
	\label{fig1}
\end{figure}

We can use the Heaviside function to produce two bounded, but discontinuous, new functions
\bb\label{5.0.4}
\Theta_1(z)=\al_1 e^{-z}H(z)\quad\text{and}\quad \Theta_2(z)=\Big( \f{c}{3}e^{2z}-\f{c}{2}e^z\Big)H(-z).
\ee

Function $\Theta_1$ is nothing but $\theta_1$ cut by $H(x)$ while $\Theta_2$ is obtained from \eqref{4.2.6} after choosing $\al_2=-c/3$, $c>0$, and then cutting it by $H(-x)$, see figures \ref{subfigA} and \ref{subfigB}.

We can ``create'' a new function $\Theta$ ``satisfying'' \eqref{5.0.1} for both $x<0$ and $x>0$ by adding $\Theta_1$ and $\Theta_2$, that is, $\Theta=\Theta_1+\Theta_2$. Moreover, we can make $\Theta$ continuous at $z=0$ by choosing $\al_1=-c/6$, namely,
\bb\label{5.0.5}
    \Theta(z)=-\f{c}{6} e^{-z}H(z)+\Big( \f{c}{3}e^{2z}-\f{c}{2}e^z\Big)H(-z).
    \ee

The steps we took for obtaining \eqref{5.0.5} from \eqref{4.2.5} and \eqref{4.2.6} are shown the Figure \ref{fig1}.

We can go even further and conclude that $\Theta$ is continuously differentiable ($C^1$) by noticing the following: first, $\Theta_1$ and $\Theta_2$ are $C^\infty$ for $z>0$ and $z<0$, respectively. Second, we have 
$$
\Theta'(\epsilon)=\f{c}{6}e^{-\epsilon}\quad\text{and}\quad\Theta'(-\epsilon)=\f{2c}{3}e^{-2\epsilon}-\f{c}{2}e^{-\epsilon},
$$
for any $\epsilon>0$, which implies the continuity of $\Theta'(\cdot)$ at $z=0$, since
$$
\Theta'(0^+)=\lim_{\epsilon\rightarrow0}\Theta'(\epsilon)=\f{c}{6}=\lim_{\epsilon\rightarrow0}\Theta'(-\epsilon)=\Theta'(0^-).
$$
We cannot go beyond $C^1$ regularity since \eqref{5.0.5} has a jump at $z=0$, see Figure \ref{fig3}.
\begin{figure}[h!]
	\centering
		\includegraphics[width=0.75\linewidth]{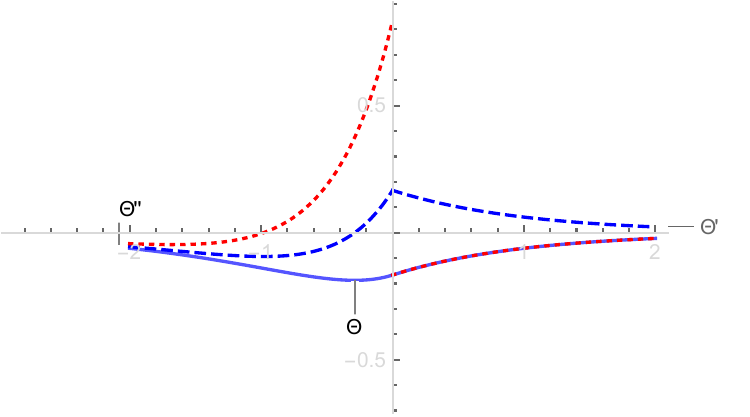}
	
	\caption{The light blue line represents the graph of the function $\Theta$ \eqref{5.0.5} while the dashed line shows its first derivative $\Theta'$. The peak at $z=0$ for $\Theta'$ makes evident issues at $z=0$, which is shown by the graph of the second derivative (red dotted   line).}
	\label{fig3}
\end{figure}

The new function above, emerging from the fusion between the bounded parts of $\Theta_1$ and $\Theta_2$ in such a way that continuity is preserved, somewhat behaves like a strong solution for \eqref{5.0.1} outside $z=0$. We can precisely make it a solution throughout our coming result.

\begin{theorem}\label{teo5.1}
Let $\Theta:\R\rightarrow\R$ be the function given by \eqref{5.0.5}. Then it is a weak solution for \eqref{5.0.1}.
\end{theorem}

\begin{proof}
All we need to do is to show that \eqref{5.0.5} satisfies \eqref{5.0.2}. At first glance one should expect it to have problems in view of the presence of the term $\Theta'$ in \eqref{5.0.2}. However, we observe that \eqref{5.0.5} is a $C^1$ function and, as a such, that term is well behaved and does not bring any issue at all.

Let $\varphi\in C_0^\infty$, $\Theta$ and $c$ as in \eqref{5.0.5}, and define
$$
I:=\int_\R\Big(c\Theta+8\Theta^2+2 (\Theta')^2\Big)\varphi dz-\int_\R(\Theta')^2\varphi' dz-\int_\R\Big(c\Theta+2\Theta^2\Big)\varphi''dz.
$$
Our goal is to show that $I\equiv0$. To this end, let us first fix $\epsilon\in(0,1)$ and split $I$ as
$$
I=I_{-}(\epsilon)+I_{+}(\epsilon)+I(\epsilon),
$$
where
$$
I_{-}(\epsilon)=\int_{-\infty}^{-\epsilon}\Big(c\Theta+8\Theta^2+2 (\Theta')^2\Big)\varphi dz-\int_{-\infty}^{-\epsilon}(\Theta')^2\varphi' dz-\int_{-\infty}^{-\epsilon}\Big(c\Theta+2\Theta^2\Big)\varphi''dz,
$$
$$
I_{+}(\epsilon)=\int^{\infty}_{\epsilon}\Big(c\Theta+8\Theta^2+2 (\Theta')^2\Big)\varphi dz-\int^{\infty}_{\epsilon}(\Theta')^2\varphi' dz-\int^{\infty}_{\epsilon}\Big(c\Theta+2\Theta^2\Big)\varphi''dz
$$
and
$$
I(\epsilon)=\int_{-\epsilon}^{\epsilon}\Big(c\Theta+8\Theta^2+2 (\Theta')^2\Big)\varphi dz-\int_{-\epsilon}^{\epsilon}(\Theta')^2\varphi' dz-\int_{-\epsilon}^{\epsilon}\Big(c\Theta+2\Theta^2\Big)\varphi''dz.
$$

Since $\epsilon\in(0,1)$, $\Theta\in C^1(\R)$ and $\varphi\in C^\infty(\R)$, we can take
$$
M:=\max_{x\in[-1,1]}\{|\Theta(x)|,|\Theta'(x)|,|\varphi(x)|,|\varphi'(x)|,|\varphi''(x)|\}
$$
and define 
$$
K:=2(2cM^2+13M^3).
$$
A straightforward manipulation shows that
$$|I(\epsilon)|\leq \int_{-\epsilon}^{\epsilon}\Big(|c||\Theta|+8|\Theta|^2+2 |\Theta'|^2\Big)|\varphi| dz+\int_{-\epsilon}^{\epsilon}|\Theta'|^2|\varphi'| dz+\int_{-\epsilon}^{\epsilon}\Big(|c||\Theta|+2|\Theta|^2\Big)|\varphi''|dz\leq K\epsilon,$$
showing that $I(\epsilon)\rightarrow0$ as $\epsilon\rightarrow0$. All we have to do is to show that $I_{+}(\epsilon)+I_{-}(\epsilon)$ go to $0$ as $\epsilon\rightarrow0$.

A fairly lengthy calculation shows that
$$
I_+(\epsilon)=c^2\Big(\f{5}{36}e^{-2\epsilon}-\f{1}{6}e^{-\epsilon}\Big)\varphi(\epsilon)+c^2\Big(\f{1}{6}e^{-\epsilon}-\f{1}{18}e^{-2\epsilon}\Big)\varphi'(\epsilon)\xlongrightarrow{\epsilon\rightarrow 0}-\f{c^2}{36}\varphi(0)-\f{c^2}{9}\varphi'(0)
$$
and
$$
\ba{lcl}
I_-(\epsilon)&=&\ds{c^2\Big(-\f{1}{2}e^{-\epsilon}+\f{17}{12}e^{-2\epsilon}-\f{4}{3}e^{-3\epsilon}+\f{4}{9}e^{-4\epsilon}\Big)\varphi(-\epsilon)}\\
\\
&&\ds{+c^2\Big(\f{1}{2}e^{-\epsilon}-\f{5}{6}e^{-2\epsilon}+\f{2}{3}e^{-3\epsilon}-\f{2}{9}e^{-4\epsilon}\Big)\varphi'(-\epsilon)\xlongrightarrow{\epsilon\rightarrow 0} \f{c^2}{36}\varphi(0)+\f{c^2}{9}\varphi'(0),}
\ea
$$
implying that $I_{+}(\epsilon)+I_{-}(\epsilon)\rightarrow0$ as $\epsilon\rightarrow0$ and so does $I\equiv0$.
\end{proof}

Let us find the Fourier transform $\hat{\Theta}$ of the function \eqref{5.0.5}. A simple calculation shows that
$$
\ba{lcl}
\hat{\Theta}(k)&=&\ds{\f{1}{\sqrt{2\pi}}\int_\R e^{-ikz}\Theta(z)dz=\f{c}{6\sqrt{2\pi}}\int_{-\infty}^0e^{-i(k+1)z}dz}\\
\\
&=&\ds{+\f{c}{6\sqrt{2\pi}}\int^{+\infty}_0\Big(2e^{-i(k-2)z}-3e^{-i(k-1)z}\Big)dz=\f{i}{\sqrt{2\pi}}\f{1}{k^3-2ik^2+k-2i}}.
\ea
$$

Calculating $|\hat{\Theta}(k)|^2$, multiplying by $(1+k^2)^s$, $s\in\R$, and integrating the result over $\R$, we obtain
$$
\int_\R (1+k^2)^s|\hat{\Theta}(k)|^2dk=\f{1}{2\pi}\int_\R \f{(1+k^2)^s}{(k^3+k^2)^2+(2k^2+2)^2}dk.
$$

For $|k|\gg 1$, the corresponding leading terms of the polynomials appearing in the integral above are $k^{2s}$ and $k^6$. Therefore, the integral above is convergent provided that $2s-6<-1$, that is, $s<5/2$.

 We recall that given $s\in\R$, the Sobolev space $H^s(\R)$ is the set of distributions $f$ for which $(1+k^2)^{s/2}\hat{f}(k)\in L^2(\R)$. Moreover, for $s\geq t$ we have $H^s(\R)\hookrightarrow H^t(\R)$ and the comments above then show that \eqref{5.0.5} belongs to $H^s(\R)$, for any $s<5/2$. Essentially, we proved the following result.
\begin{theorem}\label{teo5.2}
    For any $\epsilon\in[0,1)$, the solution \eqref{5.0.5} for the weak formulation \eqref{5.0.2} of the equation \eqref{1.0.1} belongs to $H^{\f{3}{2}+\epsilon}(\R)$, $\epsilon\in[0,1)$.
\end{theorem}

According to the Sobolev Embedding Theorem, see \cite[Theorem 6.21]{const-book}) or \cite[Proposition 1.2, page 317]{taylor}, if $s>m+1/2$, for some integer $m\geq0$, then $f\in H^s(\R)$ implies $f\in C^m(\R)\cap L^\infty(\R)$. Combining this fact with Theorem \ref{teo5.2}, we have
\begin{corollary}\label{cor5.1}
    Solution \eqref{5.0.5} is a $C^1$ function.
\end{corollary}

We have proved that \eqref{5.0.5} is a weak solution for \eqref{5.0.1}. It is then natural to ask if it might provide a weak solution for \eqref{1.0.1}.

Equation \eqref{1.0.1} can formally be rewritten as a non-local, fully non-linear evolution equation, that is,
\bb\label{5.0.9}
u_t-4uu_x+u_x^2-\p_x\Lambda^{-2}\Big(6u^2+2u_x^2\Big)-\Lambda^{-2}u_x^2=0,
\ee
where the operators $\Lambda^{-2}$ and $\p_x\Lambda^{-2}$ are, respectively, given by the convolutions 
$$\Lambda^{-2}(f)=g\ast f\quad\text{and}\quad \p_x\Lambda^{-2}(f)=g'\ast f,$$

\bb\label{5.0.10}
g(x)=\f{e^{-|x|}}{2},
\ee
and
\bb\label{5.0.11}
g'(x)=-\sign{(x)}\f{e^{-|x|}}{2},
\ee
the latter derivative being taken in the distributional sense.

Multiplying \eqref{5.0.9} by $\varphi\in C^\infty_0([0,T)\times\R)$, where $C^\infty_0([0,T)\times\R)$ means that the set of infinitely smooth compactly supported functions defined on $[0,T)\times\R$, integrating with respect to $t$ from $0$ to $T$ and $x$ over $\R$, respectively, we obtain
\bb\label{5.0.12}
\ba{l}
\ds{\int_\R u(x,0)\varphi(x,0)dx+\int_0^T \int_\R \Big(\Lambda^{-2}u_x^2\varphi\Big)(x,t) dx dt}\\
\\
\ds{
+\int_0^{T}\int_\R\Big( u\varphi_t-2u^2\varphi_x-u_x^2\varphi-\Lambda^{-2}(6u^2+2u_x^2)\varphi_x\Big)(x,t) dxdt=0.}
\ea
\ee

A function $u=u(x,t)$ satisfying \eqref{5.0.12} is said to be a local weak solution for \eqref{5.0.9} whenever it can only be defined for some $0<T<\infty$, being called {\it global} whenever $T=\infty$. 

It is again straightforward to check that any strong solution for \eqref{1.0.1} satisfies \eqref{5.0.12} (begin local for $T<\infty$ or global otherwise). Moreover, if we define $u(x,t)=\Theta(x-ct)$, where $\Theta$ is given by \eqref{5.0.5}, we then easily check that
\bb\label{5.0.13}
u(x,t)=-\f{c}{6} e^{ct-x}H(x-ct)+\Big( \f{c}{3}e^{2(x-ct)}-\f{c}{2}e^{x-ct}\Big)H(ct-x)
\ee
is a global weak solution for \eqref{1.0.1}.

We have $u\in C^0([0,\infty), H^{\f{3}{2}+\epsilon})\cap C^1([0,\infty), H^{\f{1}{2}+\epsilon})$, $\epsilon\in[0,1)$ that, in particular, tells us that $u$ is a strong solution for \eqref{5.0.9}, but it cannot be a strong solution for \eqref{1.0.1} because it is not $C^3$ with respect to $x$.

In line with Corollary \ref{cor5.1}, we prove the following:

\begin{theorem}\label{teo5.3}
    Function \eqref{5.0.13} is a $C^1([0,\infty)\times\R)$ global strong solution for the non-local, fully non-linear equation \eqref{5.0.9}. Furthermore, it is a global weak solution for \eqref{1.0.1}.
\end{theorem}

Even though \eqref{5.0.13} is a weak solution for the original equation \eqref{1.0.1}, it is simultaneously a strong solution for the corresponding non-local evolution form of this equation. To make things even more curious, the shape of the solution has a smooth peak, see Figure \ref{subfigC}, but rather than only continuous, it is a smoother peakon, or a ``$C^1$ peakon''.

This phenomenon seems to be hardly previously observed, but not completely unseen. The paper by Li and Qiao \cite{li-jmp} is the first work mentioning such a fact, whereas \cite{qiao2018} seems to be the first reporting an explicit form for this sort of solution. The authors of \cite{li-jmp,qiao2018} coined the term {\it pseudo-peakons} to name that sort of solutions. We shall maintain the same terminology in the present work as well as the term {\it $C^1$ peakon} as used before.

\begin{remark}\label{rem5.2}
    The Heaviside function $H(\cdot)$ can be removed from \eqref{5.0.5} so that it can equivalently be rewritten as
    \bb\label{5.0.11}
    \Theta(z)=-\f{c}{2}e^{-|z|}+\f{c}{3}e^{-\f{|z|}{2}(3-\sign{(z)})},
    \ee
    that makes evident the peakon behavior of \eqref{5.0.5}.

    In spite of the presence of the $\sign$ function in \eqref{5.0.11}, it is still continuous at $z=0$. In fact, if $\epsilon>0$, then
    $$\Theta(\epsilon)=-\f{c}{6}e^{-\epsilon}\xlongrightarrow{\epsilon\rightarrow 0}-\f{c}{6}\quad\text{and}\quad \Theta(-\epsilon)=\f{c}{3}e^{-2\epsilon}-\f{c}{2}e^{-\epsilon}\xlongrightarrow{\epsilon\rightarrow 0}-\f{c}{6}, $$
    showing its continuity at $z=0$.

    The distributional derivative of $\Theta$ is given by
    $$\Theta'(z)=\f{c}{2}\sign{(z)}e^{-|z|}+\f{c}{6}(1-3\sign{(z)})e^{-\f{|z|}{2}(3-\sign{(z)})}.$$

    Once more we have the presence of the $\sign$ function in $\Theta'$, that might bring issues at $z=0$. However, if $\epsilon>0$, then
    $$\Theta'(\epsilon)=\f{c}{2}e^{-\epsilon}-\f{c}{3}e^{-\epsilon}\xlongrightarrow{\epsilon\rightarrow 0}\f{c}{6}\quad\text{and}\quad \Theta'(-\epsilon)=-\f{c}{2}e^{-\epsilon}+\f{2}{3}ce^{-2\epsilon}\xlongrightarrow{\epsilon\rightarrow 0}\f{c}{6},$$
    and thus, $\Theta'$ is still continuous.
\end{remark}

\begin{remark}\label{rem5.3}
  Due to the presence of the $\sign$ function in \eqref{5.0.11}, we would expect the presence of a delta function $\delta$ in $\Theta'$. However, we observe that $(|z|\sign{(z)})'=1$. 
\end{remark}

\section{Generalisations of the pseudo-peakon solution}\label{sec6}

Solution \eqref{5.0.5}, found by the collage process, leads us wonder if we might extend it to a more general solitary weak wave solution for \eqref{1.0.1}. Let us consider the ansatz
\bb\label{6.0.1}
u(x,t)=p_1(t)e^{-x+q(t)}H(x-q(t))+\Big(p_2(t)e^{x-q(t)}+p_3(t)e^{2(x-q(t))}\Big)H(q(t)-x),
\ee
where $q,p_1,p_2$ and $p_3$ are $C^1$ functions to be determined and $H$ denotes again the Heaviside step function, and seek for a conservative solution. By conservative we mean a solution having some invariant conserved quantity. 

Taking the distributional derivatives of \eqref{6.0.1}, we obtain
\bb\label{6.0.2}
\ba{lcl}
u_t(x,t)&=&p_1'(t)e^{-x+q(t)}H(x-q(t))+\Big(p_2'(t)e^{x-q(t)}+p_3'(t)e^{2(x-q(t))}\Big)H(q(t)-x)\\ \\
&& + q'(t)p_1(t)e^{-x+q(t)}H(x-q(t)) -q'(t)(p_2(t)e^{x-q(t)} + 2p_3(t)e^{2(x-q(t))})H(q(t)-x)\\ \\
&&-q'(t)p_1(t)\delta(x-q(t))+q'(t)(p_2(t)+p_3(t))\delta(q(t)-x),
\ea
\ee
and 
\bb\label{6.0.3}
\ba{lcl}
u_x(x,t)&=&\ds{(p_1(t)-p_2(t)-p_3(t))\delta(x-q(t))}\\
\\
&&\ds{-p_1(t)e^{-x+q(t)}H(x-q(t))+\Big(p_2(t)e^{x-q(t)}+2e^{2(x-q(t))}p_3(t)\Big)H(q(t)-x)},
\ea
\ee
where $\delta$ denotes the Dirac delta function, whose presence makes the term $u_x^2$ in \eqref{5.0.9} ill defined in view of the product of the distribution. We can overcome this problem by assuming 
\bb\label{6.0.4}
p_1(t)=p_2(t)+p_3(t)
\ee
and removing it from \eqref{6.0.2} and \eqref{6.0.3}.

It is worth mentioning that condition \eqref{6.0.4} makes function \eqref{6.0.1} continuous along the curves $t\mapsto(q(t),t)$. Moreover, substituting \eqref{6.0.2}--\eqref{6.0.3} into \eqref{5.0.9} (or even into \eqref{5.0.12}), taking into account \eqref{5.0.10}--\eqref{5.0.11}, and proceeding with all due calculations, we conclude that \eqref{6.0.1} is a solution for \eqref{5.0.9} provided that $p_1, p_2,p_3$ and $q$ satisfy the continuous dynamical system
\bb\label{6.0.5}
\ba{l}
\ds{p_2'(t)+p_3'(t)+(p_2(t)+p_3(t))q'(t)+\f{14}{3}p_2(t)^2+9p_2(t)p_3(t)+\f{9}{2}p_3(t)^2=0},\\
\\
\ds{p_2'(t)-p_2(t)q'(t)-6p_2(t)^2-9p_2(t)p_3(t)-\f{9}{2}p_3(t)^2=0,}\\
\\
\ds{p_3'(t)-2p_3(t)q'(t)+\f{8}{3}p_2(t)^2=0},
\ea
\ee
jointly with the constraint \eqref{6.0.4}.

\begin{remark}\label{rem6.1}
    The last two equations in \eqref{6.0.5} imply that if $p_2p_3\equiv0$, then $p_2=p_3=0$, that jointly with \eqref{6.0.4} yield $u\equiv0$. Therefore, we assume $p_2p_3\not\equiv0$.
\end{remark}

We want a $C^1$ solution and a careful analysis on \eqref{6.0.3} tells us that it must be continuous along the curve $t\mapsto(q(t),t)$, so that we require
\bb\label{6.0.6}
-p_1(t)=p_2(t)+2p_3(t),
\ee
that, in conjunction with \eqref{6.0.4}, yield
\bb\label{6.0.7}
p_1(t)=-\f{p_3(t)}{2}\quad\text{and}\quad p_2(t)=-\f{3}{2}p_3(t).
\ee

Let us now consider the quantity \eqref{3.0.3} and assume that ${\cal H}_0(t)=-c/2$, where $c$ is a constant. By \eqref{3.0.3}, we have
\bb\label{6.0.8}
\ba{lcl}
{\cal H}_0(t)&=&\ds{\int_{q(t)}^\infty p_1(t)e^{q(t)-x}dx+\int_{-\infty}^{q(t)}\Big(p_2(t)e^{x-q(t)}+p_3(t)e^{2(x-q(t))}\Big)dx}\\
\\
&=&\ds{p_1(t)+p_2(t)+\f{p_3(t)}{2},}
\ea
\ee
and taking into account \eqref{6.0.7} and the fact that ${\cal H}_0(t)=-c/2$, we conclude that
\bb\label{6.0.9}
p_1=-\f{c}{6},\,\,\,p_2=-\f{c}{2},\,\,\,p_3=\f{c}{3}.
\ee

Calculating \eqref{3.0.4}, we get
\bb\label{6.0.10}
\ba{lcl}
{\cal H}_1(t)&=&\ds{\int_{q(t)}^\infty p_1(t)^2e^{2(q(t)-x)}dx+\int_{-\infty}^{q(t)}\Big(p_2(t)e^{x-q(t)}+p_3(t)e^{2(x-q(t))}\Big)^2dx}\\
\\
&=&\ds{\f{p_1(t)^2}{2}+\f{p_2(t)^2}{2}+\f{2}{3}p_2(t)p_3(t)+\f{p_3(t)^2}{4},}
\ea
\ee
that, jointly with \eqref{6.0.9}, implies
$$
{\cal H}_1(t)=\f{c^2}{18},
$$
meaning that it is also a conserved quantity.

Summarising, substituting \eqref{6.0.9} into \eqref{6.0.1} and taking all the comments above into account, we just proved the following.
\begin{theorem}\label{teo6.1}
    The only $C^1$ solution of \eqref{5.0.9} of the form \eqref{6.0.1} leaving \eqref{3.0.4} invariant is given by the pseudo-peakon 
    \bb\label{6.0.11}
    u(x,t)=-\f{c}{6}e^{-x+ct+q_0}H(x-ct-q_0)+\Big(\f{c}{3}e^{2(x-ct-q_0)}-\f{c}{2}e^{x-ct-q_0}\Big)H(ct+q_0-x),
    \ee
    $q_0\in\R$ and $c\neq0$.
    In particular, these two facts imply that \eqref{3.0.3} is invariant.
\end{theorem}

Theorem \ref{teo6.1} tells us that the continuity of the first derivatives of \eqref{6.0.1} implies the invariance of \eqref{3.0.4}. We might naturally wonder what happens if we no longer require a solution to be continuously differentiable, but still requiring continuity. 

Due to the exponential decay of \eqref{6.0.8} we infer that $u(\cdot,t)\in L^1(\R)$, while \eqref{6.0.10} says that $u(\cdot,t)\in L^\infty(\R)$. Similarly, from \eqref{6.0.3} we conclude that $u_x\in L^1(\R)\cap L^\infty(\R)$. Altogether, they tell us that both $u(\cdot,t)$ and $u_x(\cdot,t)$ belong to $L^r(\R)$, $1\leq r\leq \infty$, within the interval of existence of the solution. 

For this reason, we shall seek  what sort of solutions like \eqref{6.0.1} we may have whenever ${\cal H}_1'(t)\neq0$.

Without loss of generality, we can assume that ${\cal H}_0(t)=H_0$, where $H_0$ is a constant. From \eqref{6.0.8} and {\eqref{6.0.10} we obtain
\bb\label{6.0.12}
\ba{l}
\ds{p_1(t)+p_2(t)+\f{p_3(t)}{2}=H_0,}\\
\\
\ds{{\cal H}_1(t)=\f{p_1(t)^2}{2}+\f{p_2(t)^2}{2}+\f{2}{3}p_2(t)p_3(t)+\f{p_3(t)^2}{4}.}
\ea
\ee

Taking \eqref{6.0.4} into account we can eliminate $p_1(t)$ from the two equations above, and get
\bb\label{6.0.13}
p_2(t)=\frac{H_0}{2}-\f{3}{4}p_3(t)
\ee
and
\bb\label{6.0.14}
p_2(t)^2+\f{5}{3}p_2(t)p_3(t)+\f{3}{4}p_3(t)^2={\cal H}_1(t).
\ee

\begin{remark}\label{rem6.2}
    Adding the first two equations in \eqref{6.0.5} and then summing the result equation to the third multiplied by $1/2$, we obtain an equation equivalent to \eqref{6.0.13}, showing that \eqref{3.0.3} is conserved for any solution of \eqref{5.0.9} of the type \eqref{6.0.1}. 
\end{remark}

Substituting \eqref{6.0.13} in \eqref{6.0.14} we obtain
$$
p_3(t)^2+\f{4}{3}H_0p_3(t)=4(4{\cal H}_1(t)-H_0^2),
$$
whence from we conclude that $p_3(t)$ cannot be a constant as long as ${\cal H}_1'(t)\neq0$, and neither can $p_2(t)$ in view of \eqref{6.0.13}. Altogether, the conditions on $p_2(t)$ and $p_3(t)$, jointly with \eqref{6.0.5} and the condition $\mathcal{H}_0=-c/2$, provide the solution
\bb\label{6.0.15}
p_1(t)=\f{1}{6}\f{1}{t-t_0}-\f{c}{6},\,\,p_2(t)=-\f{1}{2}\f{1}{t-t_0}-\f{c}{2},\,\,p_3(t)=\f{2}{3}\f{1}{t-t_0}+\f{c}{3},\,\,q(t)=ct+q_0,
\ee
that, substituted into \eqref{6.0.1}, yield
\bb\label{6.0.16}
\ba{lcl}
u(x,t)&=&\ds{\f{1}{t-t_0}\Big(\f{e^{q_0+ct-x}}{6}H(x-ct-q_0)+\Big(\f{2e^{2(x-ct-q_0)}}{3}-\f{e^{x-ct-q_0}}{2}\Big)H(q_0+ct-x)}\Big)\\
\\
&&\ds{-\f{c}{6}e^{q_0+ct-x}H(x-ct-q_0)+\Big(\f{c}{3}e^{2(x-ct-q_0)}-\f{c}{2}e^{x-ct-q_0}\Big)H(q_0+ct-x)}.
\ea
\ee

System \eqref{6.0.5} does not have any other non-constant and non-conservative solution (in the sense that ${\cal H}_1'(t)\neq0$), meaning that $u$ given by \eqref{6.0.16} is also unique up to the choice of the parameters $q_0$, $t_0$ and $c$.

If $t_0>0$, then \eqref{6.0.16} blows up in finite time when $t$ approaches $t_0$, and its maximal interval of existence is $[0,t_0)$.  In case $t_0<0$, then it is defined for any positive time. 

We observe that $u$ is piecewise $C^1(\R\times[0,t_0))$. In fact, from \eqref{6.0.3} we infer that for $x>q(t)$ we have
$$
\lim_{x\rightarrow q(t)}u_x(x,t)=\f{c}{6}-\f{1}{6}\f{1}{t-t_0},
$$
whereas for $x<q(t)$, we have
$$
\lim_{x\rightarrow q(t)}u_x(x,t)=\f{c}{6}+\f{5}{6}\f{1}{t-t_0}.
$$

As a result, for each $t$ for which $u$ is defined, we have
$$
\lim_{n\rightarrow\infty}\Big(u_x\Big(q(t)+\f{1}{n},t\Big)-u_x\Big(q(t)-\f{1}{n},t\Big)\Big)=-\f{1}{t-t_0}.
$$

For $t_0>0$, the relation above also shows that for any $\epsilon>0$, the quantity $u_x(q(t)+\epsilon,t)-u_x(q(t)-\epsilon,t)$ is unbounded as $t$ approaches $t_0$. 

The blow up shown above is not surprising since we can easily infer that $u$ given by \eqref{6.0.16} is unbounded when $t\rightarrow t_0$ for $t_0>0$ and then, 
\bb\label{6.0.17}
\lim_{t\rightarrow t_0}\|u(\cdot,t)\|_\infty=\infty,
\ee
where $\|\cdot\|_\infty$ denotes the usual norm in $L^\infty$.
\begin{theorem}\label{teo6.2}
    Up to a choice of the constants $c$, $t_0$ and $q_0$, any weak solution for \eqref{5.0.9} of the form \eqref{6.0.1} satisfying the condition ${\cal H}_1'(t)\neq0$ and $u(x,0)\in L^{\infty}(\R)$ is piecewise $C^1$ as long as it exists. If $t_0<0$ it is a global solution, whereas for $t_0>0$ its maximal interval of existence is $[0,t_0)$ and \eqref{6.0.17} holds.
\end{theorem}

\section{Two pseudo-peakon solutions}\label{sec7}

Let us now assume a solution of the form
\bb\label{7.0.1}
u(x,t)=p_1(t)u_1(x,t)+p_2(t)u_2(x,t),
\ee
where $p_1$ and $p_2$ are $C^1$ functions and
\bb\label{7.0.2}
u_i(x,t)=-\dfrac{1}{6}e^{q_i(t)-x}H(x-q_i(t))-\dfrac{1}{2}e^{x-q_i(t)}H(q_i(t)-x)+\dfrac{1}{3}e^{2(x-q_i(t))}H(q_i(t)-x),  
\ee
$i=1,\,2$. Again, $H(\cdot)$ denotes the usual Heaviside step function and $q_1$ and $q_2$ are assumed to be $C^1$ functions satisfying the conditions
\bb\label{7.0.3}
p_1p_2\not\equiv0\quad\text{and}\quad q_1(t)\leq q_2(t).
\ee

Substituting \eqref{7.0.2} into \eqref{5.0.12} we will obtain an equation having several products involving the Heaviside step functions. We take into account \eqref{7.0.2} to derive the following relations:
\begin{enumerate}
	\item The set of functions $\left\lbrace H(q_1(t)-x),\,H(x-q_1(t))H(q_2(t)-x),\,H(x-q_2(t))\right\rbrace$ is linearly independent;
	
	\item $H(x-q_1(t))=H(x-q_1(t))H(q_2(t)-x)+H(x-q_2(t))$;
	
	\item $H(q_2(t)-x)=H(x-q_1(t))H(q_2(t)-x)+H(q_1(t)-x)$;
	
	\item $H(x-q_1(t))H(x-q_2(t))=H(x-q_2(t))$;
	
	\item $H(q_1(t)-x)H(q_2(t)-x)=H(q_1(t)-x)$;
	
	\item $H(q_1(t)-x)H(x-q_2(t))=0$.
\end{enumerate}	

Bringing this information to the equation obtained from \eqref{5.0.12} after the substitution of \eqref{7.0.1}, we obtain
\begin{align*}
&\left[\left(\dfrac{1}{3}p_1'-\dfrac{2}{3}p_1q_1'+\dfrac{2}{3}p_1^2\right)e^{-2q_1}+\left(\dfrac{1}{3}p_2'-\dfrac{2}{3}p_2q_2'+\dfrac{2}{3}p_2^2\right)e^{-2q_2}+\dfrac{4}{3}p_1p_2e^{-q_1-q_2}\right]e^{2x}H(q_1-x)+\\
&\left[\left(-\dfrac{1}{2}p_1'+\dfrac{1}{2}p_1q_1'-\dfrac{1}{2}p_1^2\right)e^{-q_1}+\left(-\dfrac{1}{2}p_2'+\dfrac{1}{2}p_2q_2'-\dfrac{1}{2}p_2^2-\dfrac{3}{2}p_1p_2\right)e^{-q_2}+\dfrac{1}{2}p_1p_2e^{q_1-2q_2}\right]e^{x}H(q_1-x)+\\
&\left[\left(-\dfrac{1}{6}p_1'-\dfrac{1}{6}p_1q_1'+\dfrac{1}{6}p_1^2+\dfrac{1}{2}p_1p_2\right)e^{q_1}+\left(-\dfrac{1}{6}p_2'-\dfrac{1}{6}p_2q_2'+\dfrac{1}{6}p_2^2\right)e^{q_2}-\dfrac{1}{6}p_1p_2e^{2q_1-q_2}\right]e^{-x}H(x-q_2)+\\
&\left[\left(-\dfrac{1}{6}p_1'-\dfrac{1}{6}p_1q_1'+\dfrac{1}{6}p_1^2\right)e^{-x+q_1}+\left(\dfrac{1}{3}p_2'-\dfrac{2}{3}p_2q_2'+\dfrac{2}{3}p_2^2\right)e^{2x-2q_2}\right]H(x-q_1)H(q_2-x)+\\
&\left[\left(-\dfrac{1}{2}p_2'+\dfrac{1}{2}p_2q_2'-\dfrac{1}{2}p_2^2\right)e^{x-q_2}+\dfrac{1}{2}p_1p_2e^{x+q_1-2q_2}-\dfrac{1}{6}p_1p_2e^{-x+2q_1-q_2}\right]H(x-q_1)H(q_2-x)=0.
\end{align*}

Equating each term multiplied by a Heaviside step function to $0$ and rearranging terms, we obtain the following dynamical system:

\begin{equation}\label{7.0.4}
	\left\lbrace\begin{aligned}
	p_1'+p_1q_1'-p_1^2+p_1p_2e^{q_1-q_2}=0,\\
	p_2'-p_2q_2'+p_2^2-p_1p_2e^{q_1-q_2}=0,\\
	p_2'-2p_2q_2'+2p_2^2=0,\\
	p_1'+2p_1q_1'-2p_1^2=0,\\
	p_1'+2p_1p_2e^{q_1-q_2}=0,\\
    p_2'-2p_1p_2e^{q_1-q_2}=0.
	\end{aligned}\right.
\end{equation}

The last two equations in \eqref{7.0.4} tell us that
\bb\label{7.0.5}
p_1+p_2=k_0,
\ee
for some constant $k_0$. Whenever $k_0=0$, we conclude from \eqref{7.0.5} that $p_1=-p_2$ and we can simplify \eqref{7.0.4} 
\begin{equation}\label{7.0.6}
\left\lbrace\begin{aligned}
p_2'+p_2q_1'+p_2^2+p_2^2e^{q_1-q_2}=0,\\
p_2(q_1'+q_2')=0,\\
p_2'-2p_2q_2'+2p_2^2=0,\\
p_2'+2p_2^2e^{q_1-q_2}=0.
\end{aligned}\right.
\end{equation}
The second equation in \eqref{7.0.6} says that $p_2=0$ or $q_1+q_2=k_1$, for another constant $k_1$. Upon substitution of $q_1=k_1-q_2$, the resulting system to be solved is
\begin{equation}\label{7.0.7}
    \left\lbrace\begin{aligned}
-p_2'+p_2q_2'-p_2^2-p_2^2e^{k_1-2q_2}=0,\\
p_2'+2p_2^2e^{k_1-2q_2}=0.
    \end{aligned}\right.
\end{equation}

Solving \eqref{7.0.7} we conclude
$$p_2(t)=\dfrac{c_1}{e^{k_1+2c_1t+2c_2}-1}\quad\text{and}\quad q_2(t)=-c_2-c_1t,$$ 
where $c_1,c_2,k_1$ are constants. Therefore, 
$$p_1(t)=-\dfrac{c_1}{e^{k_1+2c_1t+2c_2}-1}\quad\text{and}\quad q_1(t)=k_1+c_2+c_1t.$$ 

Substituting these functions into \eqref{7.0.1}, we obtain the solution

\bb\label{7.0.8}
\ba{lcl}
u(x,t)&=&\ds{-\dfrac{c_1}{e^{k_1+2c_1t+2c_2}-1}\Big(-\dfrac{1}{6}e^{k_1+c_2+c_1t-x}H(x-k_1-c_2-c_1t)}\\
\\
&&\ds{-\dfrac{1}{2}e^{x-k_1-c_2-c_1t}H(k_1+c_2+c_1t-x)+\dfrac{1}{3}e^{2(x-k_1-c_2-c_1t)}H(k_1+c_2+c_1t-x)\Big)}\\
\\
&&\ds{+\dfrac{c_1}{e^{k_1+2c_1t+2c_2}-1}\Big(-\dfrac{1}{6}e^{-c_2-c_1t-x}H(x+c_2+c_1t)}\\
\\
&&\ds{-\dfrac{1}{2}e^{x+c_2+c_1t}H(-c_2-c_1t-x)+\dfrac{1}{3}e^{2(x+c_2+c_1t)}H(-c_2-c_1t-x)\Big).}
\ea
\ee
\begin{figure}[h!]
	\centering
	\begin{subfigure}{0.45\linewidth}
		\includegraphics[width=\linewidth]{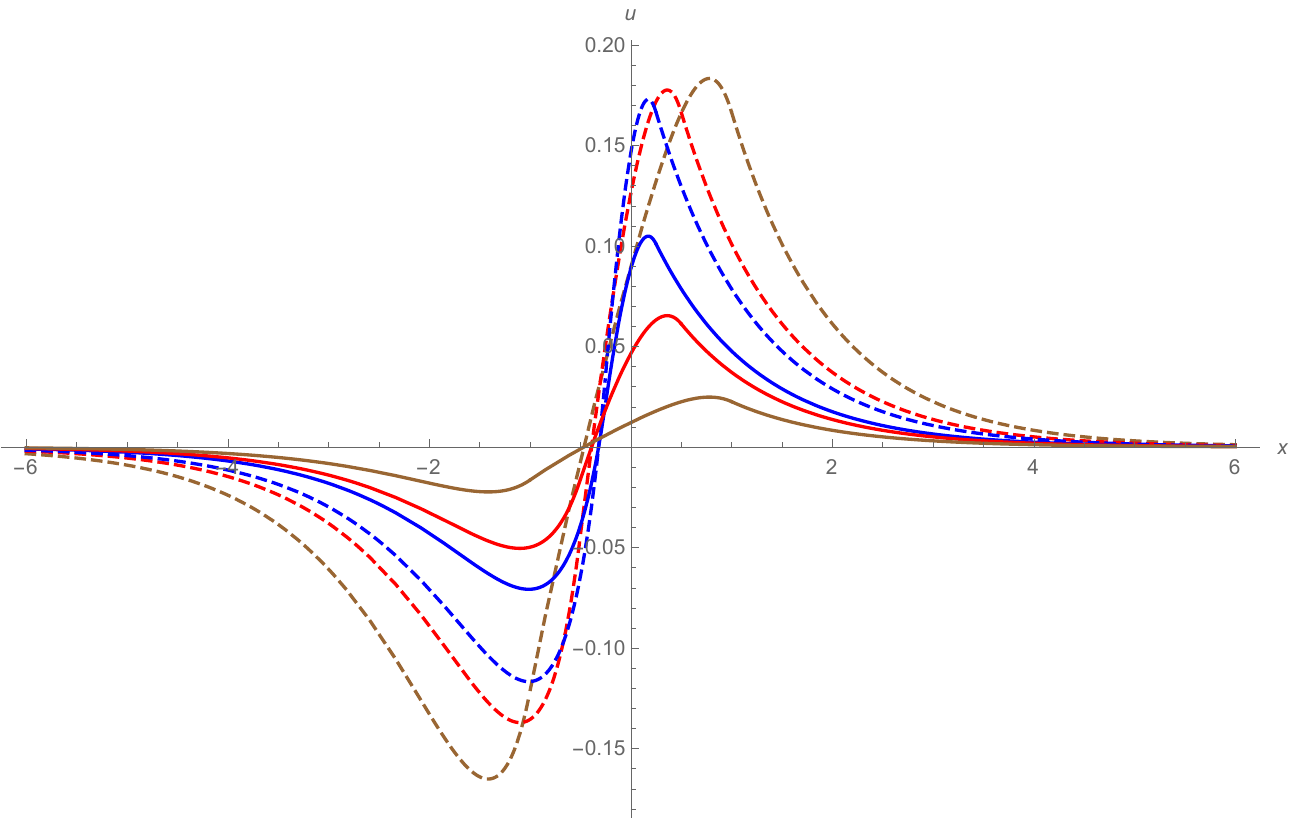}
		\caption{Solution \eqref{7.0.8} for different values of $t$: the brown, red and blue dashed lines represent $t=-1$, $t=-1/2$ and $t=-1/4$, respectively, whereas the continuous brown, red and blue lines show the cases $t=1/4,\,1/2$ and $t=1$.}
		\label{fig4a}
	\end{subfigure}\quad\quad
	\begin{subfigure}{0.45\linewidth}
	        \includegraphics[width=\linewidth]{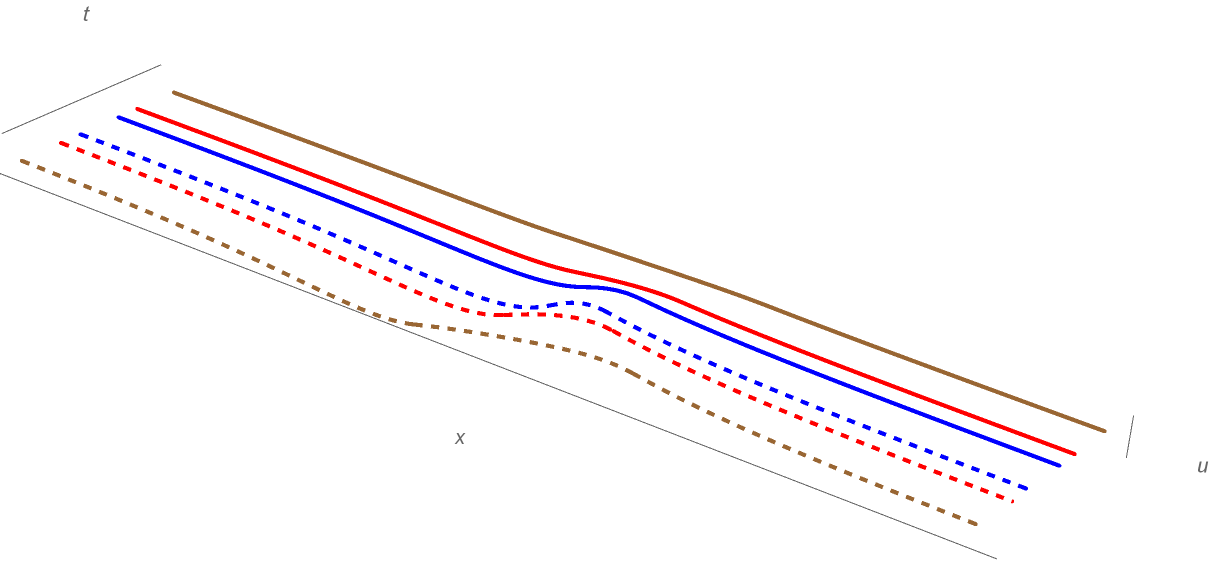}
	        \caption{Behavior of \eqref{7.0.8} for $t=\pm 1,\pm 1/2,\pm 1/4$. The colors represent the same solutions as in \ref{fig4a}.}
	        \label{fig4b}
         \end{subfigure}
	\caption{Figures \ref{fig4a} and \ref{fig4b} show \eqref{7.0.8} for $k_1=c_2=0$ and $c_1=1$. While \ref{fig4a} shows the behavior for fixed times, \ref{fig4b} show how they are evolving along time. Observe that a blow up occurs at $t=0$ and $\|u(\cdot,t)\|_{L^\infty(\R)}\leq(1-e^{2t})^{-1}$, implying that the solution vanishes as $t\rightarrow\infty$. This is also inferred from the figures, where the continuous lines represent the solutions for different values of $t>0$.}
	\label{fig6}
\end{figure}

\section{Solutions for the Degasperis-Procesi equation}\label{sec8}

Let us explore some connections between \eqref{1.0.1} and the DP equation \eqref{1.0.3}. With little effort we can rewrite \eqref{1.0.1} as
\bb\label{8.0.1}
u_t-u_{txx}=\p_x(2+\p_x)[(\p_x-2)u]^2.
\ee

Let us define a new field variable $v$ from a solution $u$ of \eqref{8.0.1} by
\bb\label{8.0.2}
v=2(\p_x-2)u.
\ee

At least formally, we have 
\bb\label{8.0.3}
u=\f{1}{2}(\p_x-2)^{-1}v,
\ee
that, substituted into \eqref{8.0.1}, yields
\bb\label{8.0.4}
v_t-v_{txx}=\f{1}{2}\p_x(\p_x^2-4)v^2,
\ee
that is nothing but the DP equation \eqref{1.0.3}. Note that solutions of the DP equation are less regular than those of \eqref{8.0.1}.

Let us give meaning to \eqref{8.0.3}. To this end, we first find the Green function $g$ of the operator $L:=2(\p_x-2)$, 
that is,
\bb\label{8.0.5}
2(\p_x-2)g=\delta(x).
\ee

Applying the Fourier transform to \eqref{8.0.5} we find
$$
\hat{g}(k)=-\f{1}{\sqrt{2\pi}}\f{1}{2(2+ik)},
$$
whereas making use of the inverse Fourier transform to the expression above, we obtain
\bb\label{8.0.6}
g(x)=-\f{1}{2}e^{2x}H(-x).
\ee

We observe that $g\in L^1(\R)\cap L^\infty(\R)$, meaning that $T_g(\phi)(x):=(g\ast\phi)(x)$ is well defined for $\phi\in H^s(\R)$. Then, we have
$$
{\cal F}\Big(T_g\phi\Big)(k)=-\f{1}{\sqrt{2\pi}}\f{1}{2(2+ik)}\hat{\phi}(k).
$$
As a consequence of the equality above, we have
$$
(1+k^2)^{s+1}\Big|{\cal F}(T_g(\phi))(k)\Big|^2=\f{1}{8\pi}\f{1+k^2}{4+k^2}(1+k^2)^s|\hat{\phi}(k)|^2
$$
and thus, $\|T_g(\phi)\|_{H^{s+1}}\leq \|\phi\|_{H^s}$,
showing the continuity of $T_g$. 

Let $\phi\in H^{s+1}(\R)$, $s>1/2$. Then $L\phi\in H^s(\R)$ and

\bb\label{8.0.7}
\ba{lcl}
T_g(L\phi)(x)&=&\ds{(g\ast(L\phi))(x)=\int_\R g(x-y)\Big(2(\p_y-2)\phi(y)\Big)dy}\\
\\
&=&\ds{2\int_\R g(x-y)\phi'(y)dy-4\int_\R g(x-y)\phi(y)dy.}
\ea
\ee

Let $\epsilon>0$ and define
$$
I(x):=2\int_\R g(x-y)\phi'(y)dy=-I_1^\epsilon-I_2^\epsilon-I_3^\epsilon,
$$
where
$$
\ba{lcl}
I_1^\epsilon&=&\ds{\int_{-\infty}^{x-\epsilon}e^{2(x-y)}H(y-x)\phi'(y)dy}
\\
\\
I_2^\epsilon&=&\ds{\int_{x-\epsilon}^{x+\epsilon}e^{2(x-y)}H(y-x)\phi'(y)dy}
\\
\\
I_3^\epsilon&=&\ds{\int_{x+\epsilon}^\infty e^{2(x-y)}H(y-x)\phi'(y)dy}
\ea
$$

Since $H(y-x)$ vanishes for $y<x$, we conclude that $I_1^\epsilon=0$, whereas the fact that $g\in L^\infty(\R)$ and $\phi'\in H^{s}(\R)$, $s>1/2$, imply $\phi'\in L^\infty(\R)$ as well. Therefore, we have
$$|I_2^\epsilon|\leq \int_{x-\epsilon}^{x+\epsilon}|g(y-x)||\phi'(y)|dy\leq 2\|g\|_{L^\infty(\R)}\|\phi\|_{L^\infty(\R)}\epsilon\xlongrightarrow{\epsilon\rightarrow 0}0.$$

In regard to $I_3^\epsilon$, integration by parts reads
$$
I^\epsilon_3=\int_{x+\epsilon}^\infty e^{2(x-y)}\phi'(y)dy=-\phi(x+\epsilon)+2\int_{x+\epsilon}^\infty e^{2(x-y)}\phi(y)dy\xlongrightarrow{\epsilon\rightarrow 0}-\phi(x)+2\int_{x}^\infty e^{2(x-y)}\phi(y)dy.
$$

Therefore, we conclude that 
$$I(x)=\phi(x)-2\int_{x}^\infty e^{2(x-y)}\phi(y)dy.$$

Substituting $I(x)$ into \eqref{8.0.7}, we get
$$
\ba{lcl}
(T_g(L\phi))(x)&=&\ds{\phi(x)-2\int_{x}^\infty e^{2(x-y)}\phi(y)dy-4\int_\R g(x-y)\phi(y)dy}\\
\\
&=&\ds{\phi(x)-2\int_{x}^\infty e^{2(x-y)}\phi(y)dy+2\int_{x}^\infty e^{2(x-y)}\phi(y)dy}=\phi(x),
\ea
$$
meaning that $T_g:H^s(\R)\rightarrow H^{s+1}(\R)$ is surjetive. In particular, since $\phi(x)=(\delta\ast\phi)(x)$ and $(Lg)(x)=\delta(x)$, we can simply write 
$$T_g(L(\phi))(x)=(g\ast(L\phi))(x)=((Lg)\ast \phi)(x)=(\delta\ast\phi)(x)=\phi(x).$$

On the other hand, if $T_g(\phi)=0$, then we have
$$
0=L(T_g\phi)=L(g\ast\phi)(x)=((Lg)\ast\phi)(x)=(\delta\ast\phi)(x)=\phi(x),
$$
and we are forced to conclude that $T_g$ is injective. Moreover, we easily conclude that $L:H^{s+1}(\R)\rightarrow H^{s}(\R)$ is bounded for $s>1/2$.

Altogether, these notes prove the following.
\begin{theorem}\label{teo8.1}
     Let $s>1/2$. Then $T_g:H^s(\R)\rightarrow H^{s+1}(\R)$ and $L:H^{s+1}(\R)\rightarrow H^s(\R)$ are bounded operators and inverses one from another. 
\end{theorem}

The non-local form of \eqref{8.0.1} is given by \eqref{5.0.9}, whereas the non-local form of the DP equation is
\bb\label{8.0.8}
v_t+vv_x+\f{3}{2}\p_x\Lambda^{-2}v^2=0.
\ee

From \eqref{8.0.2} we see that strong solutions of \eqref{5.0.9} may not be transformed into strong solutions of the DP equation \eqref{8.0.8}. Our aim in the remaining of this section is to understand the sort of solutions for the DP equation we can obtain from those for \eqref{5.0.9} we obtained in the two previous sections.

Applying the operator $L$ into the pseudo-peakon \eqref{6.0.11} we obtain
\bb\label{8.0.9}
v(x,t)=ce^{x-ct}H(ct-x)+c^{ct-x}H(x-ct)=ce^{-|x-ct|}=:v_c(x,t),
\ee
that is nothing but the well known one-peakon solution for the DP equation.
\begin{figure}[h!]
	\centering
	\begin{subfigure}{0.4\linewidth}
		\includegraphics[width=\linewidth]{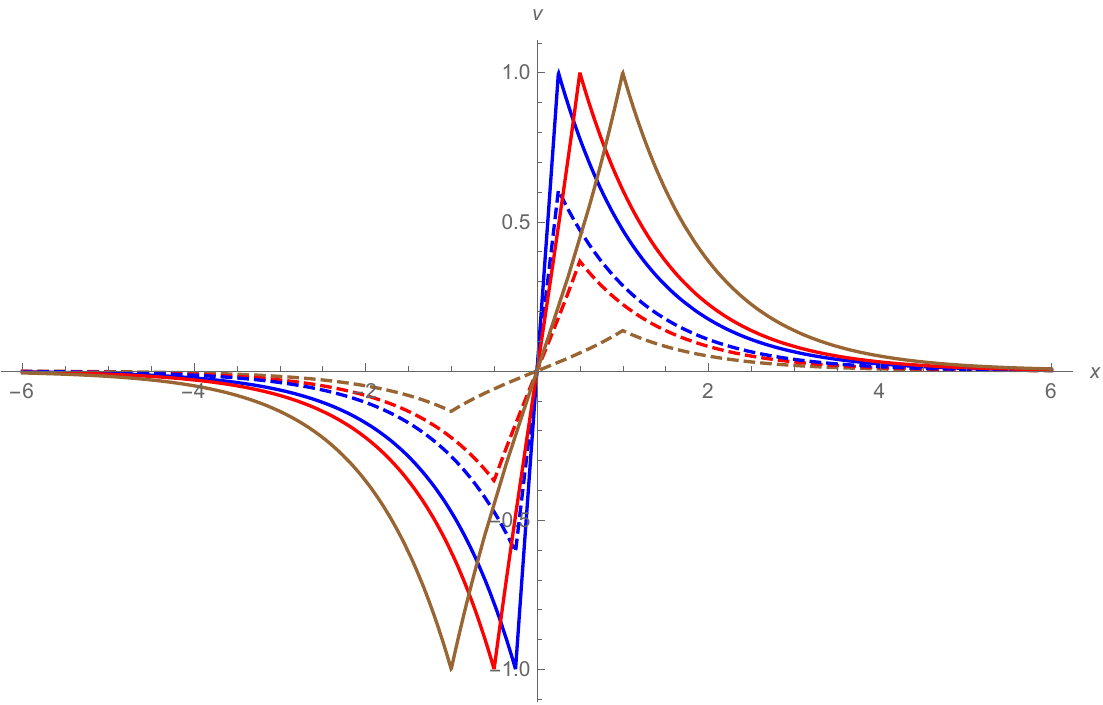}
		\caption{Solution \eqref{8.0.10} for different values of $t$: the brown, red and blue dashed lines represent $t=-1$, $t=-1/2$ and $t=-1/4$, respectively, whereas the continuous brown, red and blue lines show the cases $t=1/4,\,1/2$ and $t=1$.}
		\label{fig4a}
	\end{subfigure}\quad\quad
	\begin{subfigure}{0.4\linewidth}
	        \includegraphics[width=\linewidth]{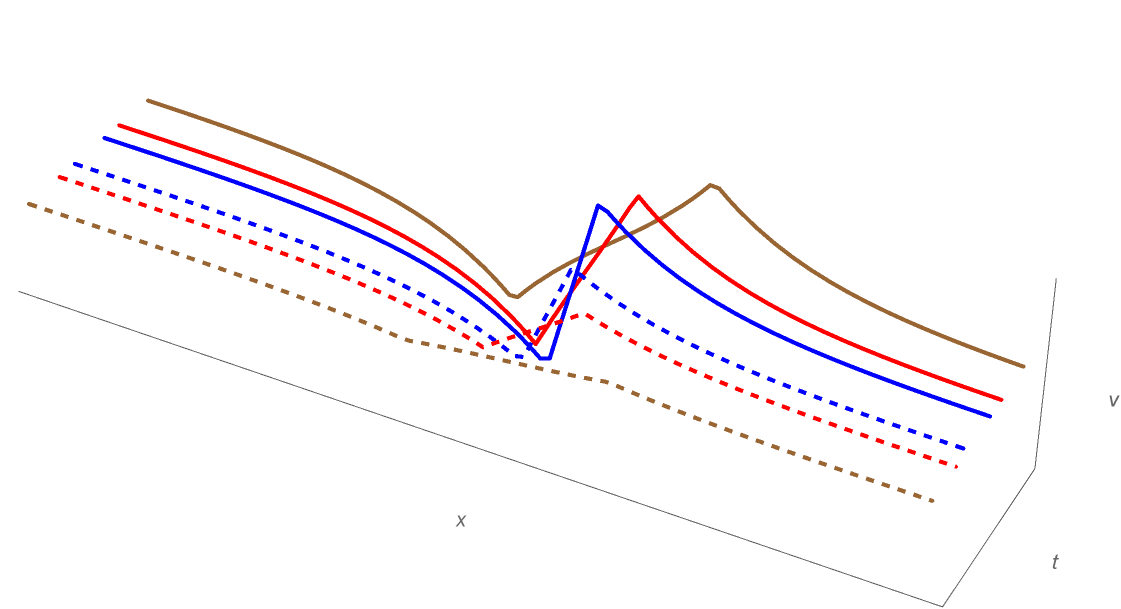}
	        \caption{Behavior of \eqref{8.0.10} for $t=\pm 1,\pm 1/2,\pm 1/4$. The colors represent the same solutions as in \ref{fig4a}.}
	        \label{fig4b}
         \end{subfigure}
	\caption{Figures \ref{fig4a} and \ref{fig4b} show \eqref{8.0.10} for $k_1=c_2=0$ and $c_1=1$. While \ref{fig4a} shows the behavior for fixed times, \ref{fig4b} show how they are evolving along time. Observe that a blow up occurs at $t=0$ and $\|u(\cdot,t)\|_{L^\infty(\R)}\leq(1-e^{2t})^{-1}$, implying that the solution vanishes as $t\rightarrow\infty$. This is also inferred from the figures, where the continuous lines represent the solutions for different values of $t>0$.}
	\label{fig4}
\end{figure}

The pseudo-peakon solution \eqref{7.0.8} provides the two-peakon solution for the DP equation
\bb\label{8.0.10}
v(x,t)=\frac{c_1}{1-e^{-2c_1t + k_1 + 2c_2}} (e^{-|c_2-k_1+c_1 t +x|}-e^{-|c_2+c_1t -x|}).
\ee

A more interesting solution is obtained by considering \eqref{6.0.16}. Although it does not satisfy the regularity conditions in Theorem \ref{teo8.1}, we can still consider $v=Lu$, where $u$ is given by \eqref{6.0.16}. Of course we cannot expect a continuous function, but the new solution of the DP equation is given by
\bb\label{8.0.11}
v(x,t)=ce^{-|x-ct|}-\f{1}{t-t_0}\sign{(x-ct)}e^{-|x-ct|},
\ee
that has to be understood in the sense of weak solutions.

If we define 
$$v_s(x,t)=-\f{1}{t-t_0}\sign{(x-ct)}e^{-|x-ct|},$$
then \eqref{8.0.11} can be written as $v(x,t)=v_c(x,t)+v_s(x,t)$, where $v_c$ is the peakon given in \eqref{8.0.9}, whereas $v_s$ correspond to a function having an evident sudden change of behavior near the discontinuity line $x=ct$. This sort of solutions for the DP equation was predicted by Coclite and Karlsen \cite{coc1,coc2} and explicitly discovered by Lundmark \cite{lund}.

Unlike peakons, that remain bounded for any $t$, solution \eqref{8.0.11} blows up in finite time for $t_0>0$, which is a feature rather different from that involving peakon solutions. As long as $t_0<0$, \eqref{8.0.11} exists for any $t>0$. As a result, for $t\gg1$, the solution approaches, or is dominated by, the peakon solution for the DP equation, that is, $v(x,t)\approx v_c(x,t)$, $t\gg1$.

\begin{figure}[h!]
	\centering
	\begin{subfigure}{0.35\linewidth}
		\includegraphics[width=\linewidth]{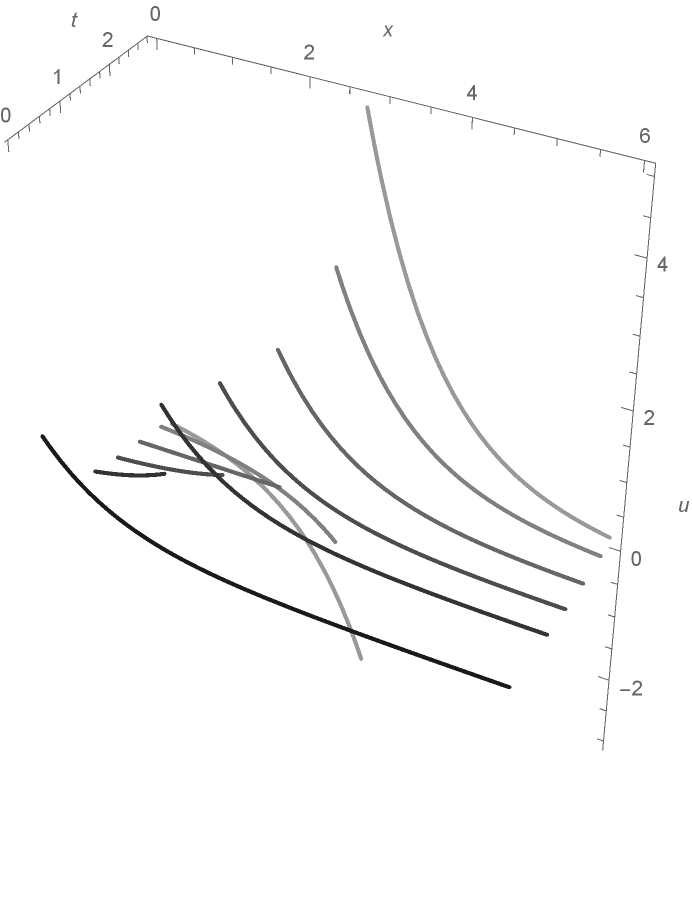}
		\caption{Solution \eqref{8.0.11} with $c=1$ and $t_0=3$. Here $t$ varies within $[0,2.75]$}
		\label{subt=3}
	\end{subfigure}\quad\quad
	\begin{subfigure}{0.45\linewidth}
	        \includegraphics[width=\linewidth]{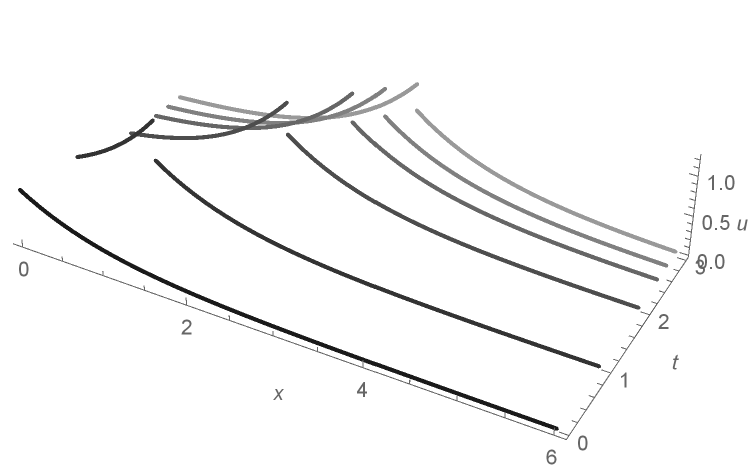}
	        \caption{Solution \eqref{8.0.11} with $c=1$ and $t_0=-3$. Here $t$ varies within $[0,3]$}
	        \label{subfict=-3}
         \end{subfigure}
	\caption{Figure \ref{subt=3} shows the behavior of \eqref{8.0.11} with $t_0=3$. The solution start increasing fast as $t$ approaches $t=3$, that is precisely its blow up. Figure \ref{subfict=-3}, on the other hand, shows \eqref{8.0.11} with $t_0=-3$. Such a choice for $t_0$ implies the global existence for $t>0$ and the solution asymptotically approaches the peakon solution for the DP equation.}
	\label{fig5}
\end{figure}

\begin{remark}\label{rem8.1}
    Solutions \eqref{8.0.11}, for $t_0>0$, as illustrated in Fig \ref{subt=3} does not qualify to be understood as an entropy solution in the sense \cite{coc1}. The situation is radically different for $t_0<0$, as the situation in Fig \ref{subfict=-3}. As long as the condition $t_0<0$ is satisfied, \eqref{8.0.11} is an entropy solution. Essentially, in the first situation the solution jumps from lower to higher values, whereas in the second case the jump discontinuity goes downwards. See \cite[Theorem 2.3]{lund} and \cite{coc2} for further details.
\end{remark}
\section{Discussion}\label{sec9}

Equation \eqref{1.0.1} was first deduced in \cite{nov} in a classification of integrable equations, where the Lax pair \eqref{1.0.2} was also reported. The existence of the Lax pair implies, in particular, the existence of an infinite number of conserved quantities. Despite their existence, very few of them have been known, most of them being obtained from the perspective of analysis of PDEs, see \cite{tu-dcds,tu-na,tu-mana}.

Our Theorem \ref{teo3.2}, on the other hand, presents six conserved currents, leading to the six conserved quantities in Theorem \ref{teo3.3}, that are obtained from the multipliers, or cosymmetries, up to second order in Theorem 3.1.

Among the quantities \eqref{3.0.3}--\eqref{3.0.7}, two are not new: the first one, that reflects the fact that the equation itself is a conservation law, and the quantity ${\cal H}_2$, see \eqref{3.0.5}, that is nothing by the square of the Sobolev norm, was first obtained in \cite[Lema 4.1]{tu-dcds}. While some of the quantities showed theorem \ref{teo3.3} are not necessarily new, the approach used to establish them are rather different from that in \cite{tu-dcds}, since we used ideas from symmetries and integrable systems, see \cite{olverbook,vitolo}, to find the cosymmetries and ultimately, the conserved quantities.

Still about symmetries, our results show that the generators of group of diffeomorphisms leaving the set of solutions invariant is spanned by four vector fields, given in \eqref{4.0.1}, whose actions on solutions are given in Table \ref{tab2}. Generators $X_1,\, X_2$ and $X_3$ correspond to translations in $x$, $t$ and a scaling in the $(t,u)-$plane, that are symmetries that \eqref{1.0.1} shares with both CH and DP equations \cite{clark}. In closing note, we observe that $X_1$, $X_2$ and $X_3$ are common generators to all equations discovered by Novikov with quadratic non-linearities \cite{nazime-tese,nazime-dcds}.

Few explicit solutions for \eqref{1.0.1} seem to have been discovered until now. The only ones we are aware of are given in \cite{tu-mana}, but that solution is not compatible with the regularity required by the coefficients of the first fundamental form \eqref{2.0.7}. The symmetries, on the other hand, naturally lead to the investigation of invariant solutions, that we carried out in Section \ref{sec4} and we used to show explicit first fundamental forms for the PSS determined by the solution of the equation. Such a geometric nature of the equation seems to have been unnoticed so far. 

The triad of one-forms defining the PSS structure of the equation does not depend on any arbitrary parameter. Moreover, relations between \eqref{1.0.1} and the DP equation, discussed in Section \ref{sec8}, jointly with the results recently reported in \cite{freire-dp}, do not allow us foreseeing \eqref{1.0.1} as a geometrically integrable equation, in the sense of \cite{reyes2002,reyes2011}, at least with a $\frak{sl}(2,\R)$ representation. This is also reinforced by the nature of the Lax pair \eqref{1.0.2}. 

Coming back to solutions, we observed that some of the solutions of ODEs obtained from symmetries have a significant change of behavior with respect to the axis $z=0$, see Figure \ref{fig1}. In particular, the families of pairs \eqref{4.2.5} and \eqref{4.2.6} have a sort of symmetry with respect to $z=0$ regarding their bounded and unbounded parts. Since for both $z>0$ and $z<0$ regions \eqref{4.2.5} and \eqref{4.2.6} are solutions for \eqref{4.2.4}, and then they provide solutions for \eqref{1.0.1}, we started wondering if we could eliminate the unbounded parts of the solutions and glue them in a smooth process. In particular, this would produce solutions compatible with some conserved quantities we had previously found.

The main problem of such an idea is just $z=0$, that could lead to a line of discontinuity. Although this would not necessarily be a problem considering conserved quantities, it might lead to a function that does not solve the equation in any sense.

In order to glue these two different solutions and yet produce a solution for the equation, we considered its weak formulation. The main reason for that is the following: any strong solution is a weak solution for the equation, but we may have weak solutions less regular than the strong one. The former could have more chances of success in surviving the surgery for removing unbounded parts of a pair of solutions, gluing them somehow so that the emerging function serves as a solution for the equation. Such a dramatic process cannot be fully executed without any price, that in our case, is the loss of regularity of the solution. The process developed in section \ref{sec5} -- cutting off unbounded parts of the solution and glue them in a continuous way -- we named {\it collage}.

Solution \eqref{5.0.5} was earlier obtained (up to notation) in \cite{tu-mana}, but in a different way. In our framework  it is a consequence of the collage process. More importantly, \eqref{5.0.5} lead us to find \eqref{6.0.16} and \eqref{7.0.8}, that are new solutions for \eqref{1.0.1}.

The collage process did not allow us to produce arbitrary new solutions. On the contrary, in order for a new solution to emerge, it selected from potential pairs those compatible with continuity, imposing strong restrictions on the original pair.

The collage process, though carried out in the present paper for \eqref{1.0.1}, can be applied to a large class of equations. In fact, we can interpret the famous peakon solutions of several non-local evolution equation, such as the $b-$equation and Novikov equation, as the resulting of the collage process applied to the solutions $u_1(x,t)=\al e^{x-ct}$ and $u_2(x,t)=\be e^{ct-x}$, $\al,\,\be\in\R$, for these equations.

Last but not least, we explored connections between \eqref{1.0.1} and the DP equation. In particular, we showed that solutions of these equations in the Sobolev class $H^s(\R)$ are in $1-1$ correspondence, at least for $s$ in the range required by Theorem \ref{teo8.1}. More interestingly, by relaxing this condition we obtained less regular solutions for the DP equation, more precisely, we got back the shock-peakon solution derived by Lundmark \cite{lund}.

We observe that \eqref{8.0.11} can be seen as a combination of a peakon solution with the discontinuous function $v_s$ (see the function given after \eqref{8.0.11}). The fact that we have a solution combining a peakon with another function like shown in \eqref{8.0.11} is not new in the literature of the DP equation. A similar fact was reported by Qiao \cite[Proposition 1]{qiao-chaos}, although Qiao's solution has a different qualitative nature. However, unlike the findings in \cite{qiao-chaos}, the solution \eqref{8.0.11} has a significant different behavior:
\begin{itemize}
    \item it develops a shock;
    \item for $t_0<0$ it is asymptotically dominated by the peakon solution outside the line of discontinuity as $t\gg1$;
    \item for $t_0\geq0$ the solution develops a blow up in finite time.
\end{itemize}

Further details on shock-peakon solutions can be found in Lundmark's paper \cite{lund}.

\section{Conclusion}\label{sec10}

In this paper we showed that equation \eqref{1.0.1} is PSS equation, in the sense its solutions endow certain regions of the $(x,t)-$space with a Riemannian metric. We found Lie symmetries of the equation, that lead us explicit solutions for the equation, as well as inspired us to develop the collage process, that produced many weak solutions for the equation, such as pseudo-peakon and shock-pseudo-peakon solution. We also established conserved quantities for the equation and showed that some of the weak solutions we obtained are compatible with them. Finally, we explored relations between the weak solutions of \eqref{1.0.1} with weak solutions of the DP equation.

\section*{Acknowledgements}

P. L. da Silva and I. L. Freire are grateful to CNPq (grants nº 310074/2021-5 and nº 308884/2022-1, respectively). All authors are thankful to FAPESP (grant nº  2024/01437-8) for financial support.

The authors would like to thank Professor Zhijun Qiao for useful discussions and point us out that \eqref{5.0.5} can be also seen as in \eqref{5.0.11}.

We are deeply indebted to Professor Hans Lundmark for encouragement and useful remarks and suggestions, that led us to improve an earlier version of the manuscript.

{\bf Declarations of interest:} none.

\end{document}